\documentclass[10pt]{article}

\usepackage{amsmath, latexsym, amsfonts, amssymb, amsthm}
\usepackage{graphicx, color,hyperref,dsfont,epsfig,caption,wrapfig,subfig}
\usepackage{datetime}
\usepackage{movie15}
\hypersetup{
    linktoc=page,
    linkcolor=red,          % color of internal links
    citecolor=blue,        % color of links to bibliography
    filecolor=blue,      % color of file links
    urlcolor=cyan,
    colorlinks=true           % color of external links
}

\numberwithin{equation}{section}

\newtheorem{theorem}{Theorem}[section]
\newtheorem{lemma}[theorem]{Lemma}
\newtheorem{proposition}[theorem]{Proposition}
\newtheorem{corollary}[theorem]{Corollary}

\newtheorem{remark}[theorem]{Remark}
\newtheorem{definition}[theorem]{Definition}

\theoremstyle{definition}

\renewcommand{\tilde}{\widetilde}          % wider `tilde'
\DeclareMathSymbol{\leqslant}{\mathalpha}{AMSa}{"36} % nicer `smaller or equal'
\DeclareMathSymbol{\geqslant}{\mathalpha}{AMSa}{"3E} % nicer `larger or equal'
\DeclareMathSymbol{\eset}{\mathalpha}{AMSb}{"3F}     % nicer `emptyset'
\renewcommand{\leq}{\;\leqslant\;}                   % redef. of < or =
\renewcommand{\geq}{\;\geqslant\;}                   % redef. of > or =
\newcommand{\C}{\mathbb{C}}

\newcommand{\D}{\mathbb{D}}
\newcommand{\R}{\mathbb{R}}
\newcommand{\Z}{\mathbb{Z}}

\newcommand{\E}{\mathds{E}}
\renewcommand{\P}{\mathds{P}}

\newcommand{\hf}{\frac{_1}{^2}}
\usepackage{xparse}

\def\bi{\begin{itemize}}
\def\ei{\end{itemize}}

\def\bnum{\begin{enumerate}}
\def\enum{\end{enumerate}}
\def\<#1{\langle #1 \rangle}

\def\z{\mathbf{z}}

\newcommand{\caF}{{\mathcal F}}

\newcommand{\caH}{{\mathcal H}}

\newcommand{\caN}{{\mathcal N}}
\newcommand{\caO}{{\mathcal O}}

\newcommand{\caR}{{\mathcal R}}

\newcommand{\caT}{{\mathcal T}}

\newcommand{\caV}{{\mathcal V}}

\title{Constructive Liouville Conformal Field Theory  \footnote
{ Lectures at the Cargese summer school  "Quantum integrable systems, conformal field theories and stochastic processes"}}

\author{ Antti Kupiainen \footnote{ Supported by the Academy of Finland}\\University of Helsinki, Department of Mathematics and Statistics\\ antti.kupiainen@ helsinki.fi}%P.O. Box 68, Helsinki,
\date{}
\begin{document}

\maketitle
 \begin{abstract}
 These lectures give an introduction to   a probabilistic approach to Liouville Quantum Field Theory developed in a joint work with F. David, R. Rhodes and V. Vargas.

 \end{abstract}
 %\vskip 5mm
 %\noindent {\it Conformal field theory is "an unsuccessful attempt to solve the Liouville theory" (A. M. Polyakov) \cite {MG}}
 %\vskip 5mm

\section{Probabilistic Liouville Theory} 

One of the simplest and at the same time most intriguing of the Conformal Field Theories (CFT) is the Liouville CFT  (LCFT). It first appeared in Polyakov's formulation of String Theory \cite{Pol} and then  in the work of Knizhnik, Polyakov and Zamoldchicov \cite {KPZ} on the relationship between CFT's in fixed background metric and in a random metric (2d gravity)  or in other terms on the relationship between statistical mechanics models on fixed lattices and on random lattices. %Indeed, the correlation functions of a CFT coupled to random metric are given as products of ordinary CFT correlations and Liouville correlations. 
Decisive progress in LCFT came in the 90's as Dorn and Otto \cite{Do} and Zamolodchikov and Zamolodchikov \cite{ZZ} produced an explicit formula for Liouville three point functions, the celebrated DOZZ-formula. More recently, the Liouville three point functions were shown  to have a deep relationship to four dimensional Yang Mills theories \cite{AGT}. 

Unlike most CFT's the LCFT has an explicit functional integral formulation. However the exact results for Liouville correlations are not derived from this functional integral but rather from general principles of CFT (BPZ equations, crossing symmetry) coupled to assumptions about the spectrum of LCFT \cite{Rib}.  A rigorous probabilistic  formulation of the LCFT was given in  \cite{DKRV}. In that work it was shown that  the LCFT functional integral can be defined using the theory of  Gaussian Multiplicative Chaos, a well studied chapter of probability theory.  In this paper I will review in Section 1 this construction as well as the proof of local conformal invariance given in\cite{krv}. In Section 2 it is shown how the Quantum Field Theory structure and a representation of the Virasoro algebra can be derived from the probabilistic theory. 

\subsection{Scaling Limits} 
One of the motivations for the study of LCFT comes from the study of scaling limits of discrete models of random surfaces. This subsection gives a brief summary of the discrete objects that LCFT aims to describe.

\subsubsection{Random triangulations} By a {\it triangulation} of the unit sphere we mean a finite connected graph  $T$ s.t. there is an  embedding of $T$ to the two dimensional sphere  ${S}^2$ s.t. each connected component of $	S^2\setminus T$ (a {\it face})  has a boundary consisting of 3 edges (we denote the embedding of $T$ by $T$ again). %We identify two graphs $G$ and $G'$ if there is an orientation preserving homeomorphism of $S^2$ mapping $G$ to $G'$ and denote by $T$ the equivalence class. 
A {\it marked } triangulation is a triangulation $T$ together with a choice of three vertices $v_1,v_2,v_3$ of $T$ . We denote by  $\mathcal{T}$  the set of marked triangulations, by $\caV(T)$ the set of  vertices of $T$ and by $|T|$ the number of faces  in $T$. 

Next we want to consider probability measures on  $\mathcal{T}$. The simplest example is the case of "pure gravity". We define the probability
\begin{align*}
\P_{{\mu_0},\sqrt{\frac{8}{3}}}(T)=\frac{1}{Z_{{\mu_0},\frac{8}{3}}}
e^{-{{\mu_0}} |T|}%Z_{{\gamma}}(T)
%\label{Pdef}
\end{align*}
(for the index  $\sqrt{\frac{_8}{3}}$, see below) where $Z_{{\mu_0},\sqrt{\frac{8}{3}}}=\sum_{T\in\caT}
e^{-{{\mu_0}} |T|}$. For other examples we add "matter" to the gravity model.
Given a triangulation $T$ one may consider statistical mechanics models on it. For example, for the Ising model one defines "spin" variables $\sigma_v\in\{1,-1\}$ indexed by the vertices $v\in\caV(T)$ of $T$ and considers the joint probability distribution %$\P_\beta$ 
on triangulations $T$ and spin configurations $\sigma=\{\sigma_v\}_{v\in\caV(T)}$ defined by
$$
\P_{{\mu_0},\sqrt{3}}(T,\sigma)=\frac{1}{Z_{{\mu_0},\sqrt{3}}}
e^{-{{\mu_0}} |T|}
e^{\beta_c\sum_{v\sim v'}\sigma_v\sigma_{v'}}
$$ where $v\sim v'$ means $v,v'$ share an edge. The parameter $\beta_c$ is the critical value for the inverse temperature of the Ising model, known to exist in this setup.  More generally there are many other critical statistical mechanical models one can define on $T$. Their marginal distributions on $\caT$ are all of the form
%We will consider a two-parameter family of probability measures $\P_{{\mu_0},\gamma}$ on 
%For $\mu_0>0$ consider the probability $\P_{\mu_0}$ on 
%$\mathcal{T}$ %=\cup_N\mathcal{T}_{N}$ 
%defined by
\begin{align}
\P_{{\mu_0},\gamma}(T)=\frac{1}{Z_{{\mu_0},\gamma}}
e^{-{{\mu_0}} |T|}Z_{{\gamma}}(T)
\label{Pdef}
\end{align}
Here  $Z_{{\gamma}}(T)$ is the partition function of  the statistical model %a {\it critical lattice model}
on the graph $T$ and $\gamma$ is a parameter depending on that model. It is related to its central charge $c$ by $c=25-6Q^2$, $Q=\frac{\gamma}{2}+\frac{2}{\gamma}$. In particular $\gamma\in [\sqrt{2},2]$ corresponding to $c\in [-2,1]$. Examples of models covering this whole range are given by the $O(N)$  loop models. The  $\gamma=2$ case is discrete Gaussian Free Field where   $Z_{{\gamma}}(T))=\det (-\Delta_T)^{-\hf}$ with $\Delta_T$ the Laplacean on the graph $T$ and the   $\gamma=\sqrt{2}$ case is the uniform spanning tree with  $Z_{{\gamma}}(T))=\det (-\Delta_T)$. 

 It is known that
\begin{equation}
Z_N:=\sum\limits_{T \in \mathcal{T}: |T|=N} Z_{\gamma}(T)=N^{1-\frac{4}{\gamma^2}}e^{\bar \mu N}(1+o(1))\label{ASY}
\end{equation}
where $\bar\mu$ depends on the model. In particular this implies that  $\P_{\mu_0,\gamma}$ is defined for ${\mu_0}>\bar \mu$ and 
%$$\lim_{{\mu_0}\downarrow\bar\mu}\E_{\mu_0,\gamma} |T|=\infty.
$\lim_{{\mu_0}\downarrow \bar\mu}Z_{\mu_o,\gamma}=\infty$ if $\gamma\geq\sqrt 2$.
Hence as $\mu_0\to\bar\mu $ the measure samples  large triangulations.

\subsubsection{Conformal structure}
For each $T$ we may associate a conformal structure on $S^2$ as follows. Assign to each face $f$ a copy $\Delta_f$ of an equilateral triangle $\Delta$ of unit area and let $M_T=\sqcup \Delta_f/\sim$ where in  the disjoint union of  the $\Delta_f$ we identify the common edges. $M_T$ is a topological manifold homeomorphic to $S^2$. 

We can make $M_T$ a complex manifold by the following atlas. It consists of the following coordinate patches. First,  interiors of $\Delta_f$ are  mapped by identity to $\Delta$. Second, for each pair of faces $f$ and $f'$ that share an edge we map the interiors of $\Delta_f\cup\Delta_{f'}$ by identity to two copies of the standard triangle $\Delta$ sitting next to each other in $\C$. Finally for each vertex $v\in M$  we map its neighbourhood  to $\C$ as follows. First, list the faces sharing $v$ in consecutive order: $f_0,\dots,f_{n-1}$. Then  parametrize the set $\Delta_{f_j}\cap U$  by $z_j=re^{2{\pi} i\theta_j}$ with $\theta_j\in [6j/n,6(j+1)/n]$. Then $z\to z^{n/6}$ provides a complex coordinate for  a neigborhood of $v$. This atlas makes $M_T$ a complex manifold homemorphic to $S^2$. Picking three points $z_1,z_2,z_3\in \C$ there is a unique conformal map $\psi_T: M_T\to\hat\C$ s.t. $\psi(v_i)=z_i$ where $\hat\C$  is the Riemann sphere, see Section \ref{secKPZ}.

Let $\lambda_T$ be the area measure on $M_T$   i.e.  $\lambda_T$ is the Lebesque measure in the local coordinate on $\Delta_f$.  Let $\gamma_T$ be the Riemannian  metric on $M_T$  which in  the local coordinate on $\Delta_f$ is given by $dx\otimes dx+dy\otimes dy$. We may transport these objects to $S^2$ by the conformal map $\psi_T$.  %Let us denote by $\nu_T$ %=\psi_\ast\lambda_T$ 
%and $g_T$  the resulting measure and metric on the sphere. 
If we now sample $T$ 
%Under $\P_{{\mu_0},\gamma}$,  $\nu_{T}$ becomes a random measure $\nu_{{{\mu_0}},\gamma}$ on ${S}^2$. $d_T$ the induced metric on $S^2$. 
from $\P_{{\mu_0},\gamma}$, these become  a random measure $\nu_{{{\mu_0}},\gamma}$ and a random Riemannian metric $G_{{{\mu_0}},\gamma}$ on $\hat\C$. In the standard coordinate of $\hat\C$ they are given by  $\nu_{{{\mu_0}},\gamma}=g_{{{\mu_0}},\gamma}(z)dz$\footnote{we denote the Lebesque measure on $\C$ by $dz$} and $G_{{{\mu_0}},\gamma}=g_{{{\mu_0}},\gamma}(z)(dx\otimes dx+dy\otimes dy)$ where the density $g_{\mu_0,\gamma}$ is singular at the images of the vertices with $n>6$. 

\subsubsection{Scaling limit}
Consider now a {\it scaling limit} as follows. Recalling that as ${\mu_0}\downarrow\bar\mu$ typical size of triangulation diverges we define
for $\mu>0$ 
\begin{align*}
\rho^{(\epsilon)}_{\mu,\gamma}:=\epsilon\nu_{{\bar\mu+\epsilon\mu},\gamma}\ \ \ \
g^{(\epsilon)}_{\mu,\gamma}:=\epsilon^{a_\gamma}\Gamma_{{{\mu_0}},\gamma}
\end{align*}
Then its is conjectured that  $\rho^{(\epsilon)}_{\mu,\gamma}$ converges as $\epsilon\to 0$ to a random measure $\rho_{\mu,\gamma}$ on $\hat\C$ and the metric space defined by $g^{(\epsilon)}_{\mu,\gamma}$ %$\epsilon^{a_\gamma}\Gamma_{{{\mu_0}},\gamma}$  and metric $g_{\mu,\gamma}$ on $\hat\C$. 
converges to a random metric space.
In the case $\gamma=\sqrt{8/3}$ $a_\gamma=\hf$ and this was proven % in the setup of random metric space convergence by 
Le Gall \cite{LeGall} and Miermont \cite{Mier} and the random metric was constructed directly in the continuum by Miller and Sheffield \cite{MS} .
Since $\epsilon\nu_{T}({S}^2 )=\epsilon N$ the asymptotics \eqref{ASY} implies that the law of $\rho_{\mu,\gamma}({S}^2 ) $ is given by
$$
\E[  F(\rho_{\mu,\gamma}({S}^2 )   )  ]=\lim_{\epsilon\to 0} \frac{1}{Z_{\epsilon}}\sum_N e^{-\mu \epsilon N}N^{1-\frac{4}{\gamma^2}}
F(\epsilon N )
$$
%$$
%\E[  F(\rho^{(\epsilon)}_{\mu,\gamma}({S}^2 )   )  ]= \frac{1}{Z_{\epsilon}}\sum_N e^{-\mu \epsilon N}N^{1-\frac{4}{\gamma^2}}
%F(\epsilon N ).
%$$
%\eqref{ASY} implies this law converges to 
i.e. the law is $\Gamma(2-\frac{4}{\gamma^2},\mu ) $. In what follows we will construct a random measure that has this law  for its total mass and is a candidate for the scaling limit. 
 %In \cite{DKRV}
 
 As another example of a limiting object consider the case of Ising model ($\gamma=\sqrt 3$). We can transport the Ising spins $\sigma_v=\pm 1$ sitting at vertices $v$ of $T$ to $\hat\C$. Define the %random 
 distribution
 \begin{align}\label{ising}
 \Phi_T^{(\epsilon)}(z)=\epsilon^{\frac{5}{6}}\sum_{v\in{\cal V}(T)}\sigma_v\delta(z-\psi_T(v)).
 \end{align}
Then under $\P_{\mu_0+\epsilon\mu,\gamma}$ this becomes a random field on $\hat\C$ and we will get  a conjecture for its distribution as $\epsilon\to 0$ in terms of the correlation functions of  the Liouville QFT, see Section \ref{lme}.

\subsection{KPZ Conjecture} \label{secKPZ}

Locality and coordinate invariance are the basic principles of relativistic physics. Locality means that the basic objects are fields that are functions on the space-time manifold  $M$ (string theory is an exception to this) and their dynamics is determined by an action functional that is local in the fields and their derivatives, e.g. the free scalar field has
$$
S(\phi)=\int_M (\nabla\phi)^2dx.
$$  
(General) relativity enters also through a local field, (pseudo) Riemannian metric $g(x)=\sum g_{\alpha\beta}(x)dx^\alpha\otimes dx^\beta$. In (Euclidean) quantum gravity one looks for a probability law in the space of fields.  Coordinate invariance means that this law should be invariant under coordinate transformations i.e. under  the action of the group of diffeomorphisms $Diff(M)$. Hence in particular this law lives on the space of metrics modulo diffeomorphisms $Met(M)/Diff(M)$. In two dimensions this space is particularly simple. In particular on the sphere $S^2$ any two smooth metrics $g,g'$  are, modulo a diffeomorphism, conformally equivalent, i.e.   $f^\ast g'=e^\varphi g$ for $\varphi:M\to\R$. %On a compact two dimensional manifold the set of conformal classes is a finite dimensional space and in particular on the sphere $S^2$ there is only one conformal class.

\subsubsection{Conformal metrics} Recall that the Riemann sphere $\hat\C=\C\cup\{\infty\}$ can be covered by two coordinate patches $\hat\C\setminus\{\infty\}$ and $\hat\C\setminus \{0\}$ with the coordinates %transition function 
$z$ and $z^{-1}$. A  conformal metric on $\C$   is given  by $\hf g(z)(dz\otimes d\bar z+d\bar z\otimes dz)$%:=g(z)(dz)^2$ 
which becomes $g(1/z)|\bar zz|^{-2}$ on the other patch. Hence  if $g$ is  continuous  on  $\hat\C$  this means  $g(z)=\caO(|\bar zz|^{-2})$ at infinity. The round metric is given by %the one invariant under $z\to z^{-1}$:
\begin{align}
\hat g(z)=4(1+|\bar zz|)^{-2}
\label{hat g}
\end{align}
and it has the area is $\int_\C\hat g=4\pi$ and the scalar curvature $R_g:=-4g^{-1}\partial_{\bar z}\partial_z\ln g$ is  constant: $R_{\hat g}=2$. For all smooth conformal metrics one computes $\int gR_g=\hf i\int\bar\partial\partial\ln g=8\pi$, an instance of the Gauss-Bonnet theorem. %=i\lim_{R\to\infty}\oint_{|z|=R}\bar\partial\log g=i\lim_{R\to\infty}\oint_{|z|=R} z^{-1}=2\pi$.
Given a conformal metric on $\hat\C$ we can define the Sobolev space $H^1(\hat \C,g)$ with the norm 
$$\|f\|_g^2:=\int(|\partial_zf|^2+g(z)|f|^2)dz.$$
 These norms are equivalent for all continuous conformal metrics. Finally we define $H^{-1}(\hat \C,g)$ as the dual space and denote the dual pairing by $\langle X,f\rangle$. Formally  $\langle X,f\rangle=\int X(z)f(z)g(z)dz$.

\subsubsection{Liouville QFT}

For (Euclidean) quantum gravity on $\hat \C$ one is thus looking for the probability law of the conformal metric $\hf e^\varphi(dz\otimes d\bar z+d\bar z\otimes dz)$ i.e. for a law for a random real valued field $\varphi$. To state the KPZ conjecture for this law we fix a conformal metric $g(z)$ ("background metric")  on  $\hat\C$ and then the  KPZ conjecture  \cite{KPZ, Da, DiKa} states that  the random measure  $\rho_{\mu,\gamma}$ is given by
\begin{align}
\rho_{\mu,\gamma}(dz)=e^{\gamma \phi_g(z)}%g(z)^{1+\frac{\gamma^2}{4}}
dz\label{rho}
\end{align}
where $ \phi_g$ is
 the {\it Liouville field}
\begin{align}\label{Liouville field}
\phi_g:=X+\frac{_Q}{^2}\ln g
%\label{actio}
\end{align}
and
$X$ is a random field whose law  is  formally given by 
\begin{align}
\E_{\gamma,\mu}\,  f(X)=Z^{-1}\int_{Map(\C\to\R)} f(X)\, e^{-S_L(X,{g})}DX.
\label{funci}
\end{align}
Here $S_L$ is action functional of the {\it Liouville model}: 
\begin{align}
S_L(X,{g}):= \frac{1}{\pi}
\int_{\C}\big(\partial_z X\partial_{\bar z} X+\frac{_Q}{^4}gR_{{g}} X  +\pi\mu e^{\gamma \phi_g  }\big)\,dz.
\label{actio}
\end{align}
 Here %$g(z)dz$ is the volume measure of a smooth conformal metric on the Riemann sphere and 
 %$R_g=-4g^{-1}\partial_z \partial_{\bar z}\log g
%$ is the scalar curvature. 
$Q$ is related to $\gamma$ by
$$Q=2/\gamma+\gamma/2.$$
Furthermore the heuristic integration over $X$ in \eqref{funci} is supposed to include "gauge fixing" due to the marked points $z_1,z_2,z_3$. 
{Our aim} is to give precise meaning to the law \eqref{funci} and study its properties that include 
{\it conformal invariance}.
\begin{remark}
Note that for  a conformally equivalent metric $g'=e^\varphi g$
$$
g'R_{g'}=gR_g-4\partial_{\bar z}\partial_z\varphi
$$
so that
\begin{align}
S_L(X,{e^\varphi g})=S_L(X+\frac{_Q}{^2}\varphi,{g})- \frac{Q^2}{4\pi}\int_{\C}\big(\partial_z \varphi\partial_{\bar z} \varphi+\hf gR_{{g}} \varphi)dz
\label{weyl0}
\end{align}
Thus, modulo an additive constant, the {\it Weyl transformation} $g\to e^\varphi g$ is a shift in $X$. %One might be tempted to 
\end{remark}

\begin{remark}
 $\hat\C$ has a nontrivial automorpism group $SL(2,\C)$ which acts as M\"obius transformations $\psi(z)=\frac{az+b}{cz+d}$. By change of variables one can compute
\begin{align*}%\label{eq:liouville}
S_L(X\circ\psi^{-1},g)=S_L(X+\frac{_Q}{^2}\varphi,g)
\end{align*}
where $e^\varphi g=|\psi'|^2%\frac
{ g\circ\psi}$. 
\end{remark}

\subsection{Massless Free Field} Let us first keep only the quadratic term in the action functional \eqref{actio} and try to define the linear functional
\begin{align}
\langle F\rangle=\int_{Map(\C\to\R)} F(X)e^{-\frac{1}{\pi}
\int_{\C}|\partial_zX |^2dz}%\lambda_{g}}
DX
\label{gffexp}
\end{align}
We may define this in terms of the Gaussian Free Field (GFF). 

%\subsubsection{Random fields} We will view a random field $X$ on $\C$   in two ways: first, as a random distribution $X\in H^{-1}(\hat \C,g)$ which means that for $f\in H^{1}(\hat \C,g)$  $\langle X,f\rangle$ are random variables defined on some probability space or,  second,  as a  probabilty measure $\nu$ on $H^{-1}(\hat \C,g)$. Such fields are determined by their generating function 
%\begin{align}
%W(f)=\E\, e^{i\langle X, f\rangle}=\int e^{i\langle X, f\rangle}\nu(dX).
%\label{rffexp}
%\end{align} 
%which is continuous on $\caS(\C)$ of positive type. Conversely such a $W(f)$ (satisfying $W(0)=1$) determines a unique measure $\nu$. 
%A  {\it Gaussian measure} on $H^{-1}(\hat \C,g)$ has $W(f)=e^{-\hf C(f,f)}$ where  the {\it covariance} $C(f,g)=\E \langle X, f\rangle \langle X,f \rangle $  is a continuous non-negative bilinear form on $H^{1}(\hat \C,g)\times H^{1}(\hat \C,g)$. Examples are given by suitable functions $C(z,w)$ with $C(f,g)=\int C(z,w)f(z)g(w)g(z)g(w)dzdw$. 

\subsubsection{GFF} In general  the $GFF$ is a Gaussian random field whose covariance is the Green function of the Laplacean. In our setup the Laplace operator is given by  $\Delta_g=4g(z)^{-1}\partial_{\bar z}\partial_{ z}$. Some care is needed here since $\Delta_g$ is not invertible. Indeed, $-\Delta_g$ is a non-negative self-adjoint operator on $L^2(\hat\C,g)$ (whose inner product we denote by $(f,h)_g=\int \bar fhgdz$). It has a point spectrum consisting of eigenvalues $\lambda_n$ and  orthonormal eigenvectors $e_n$ which we take so that $\lambda_n>0 $ except for $\lambda_0=0$ with $e_0=1/|1\|_g$. We define the GFF $X_g$ as the random distribution
\begin{align}
X_g(z)=\sqrt{2\pi}\sum_{n>0}\frac{x_n}{\sqrt{\lambda_n}}e_n(z)
\label{gff1}
\end{align}
where  $x_n$ are i.i.d. $N(0,1)$. The covariance  
$G_g(z,z'):=\E X_g(z) X_g(z')$ is easily computed: we have (for real $f$, $h$)
$$
\frac{1}{2\pi}\E( X_g,(-\Delta_g)f )_g(X_g,h)_g=(f,h)_g-(e_0,f)_g(e_0,h)_g
$$
which implies that 
\begin{align}\label{deltag}
-\Delta_g G_g(z,z')=2\pi(g(z)^{-1}\delta(z-z') -(\int g(w)dw)^{-1})
\end{align}
Since $(e_0,X_g)_g=0$  we have $\int G_g(z,z')g(z)dz=0=\int G_g(z,z')g(z')dz'$ and we end up with
\begin{align}\label{covg}
G_g(z,z')=\E X_g(z)X_g(z')=\ln|z-z'|^{-1} %m_g(\ln|z-\cdot|^{-1}-m_g(\ln|z'-\cdot|^{-1}+m_g\times m_g\ln|\cdot-\cdot'|^{-1}
-c_g(z)-c_g(z')+C_g
%\log(\int_\C g(v)dv )^{-2}\int_{\C^2}\ln\frac{|z-u||z'-u'|}{|z-z'||u-u'|}g(u)g(u')dudu'.
\end{align}
where 
\begin{align}\label{covg1}
c_g(z)=m_g(\ln|z-\cdot|^{-1})=\ln|z|+\caO(1)=\frac{_1}{^4}\ln {g}(z)+\caO(1)
\end{align}
 and $C_g=(1,g)^{-2}\int \ln|u-v|^{-1}g(u)g(v)dudv$. We used the notation for the average in $g$
 $$
m_g(f):=\int_\C f(z)\, g(z)dz/\int_\C g(z)dz
$$
 For the round metric we have 
 \begin{align}\label{covg2}c_{\hat g}=\frac{_1}{^4}\ln {\hat g}-\hf\ln 2, 
\ \ \ C_{\hat g}=-\hf.\end{align}
One should think about the  $X_g$ as we vary $g$ as obtained from the same field $X$ by
 $X_g=X-m_g(X)$. Although there is no such $X$ this makes the following fact evident. If $g'$ is another conformal metric then
\begin{align}\label{change}
X_{g'}\stackrel{law}{=}X_{g}-m_{g'}(X_g).
\end{align}
Moreover
the GFF $X_g$ transforms simply under % {\it M\"obius covariant:} let $\psi$ be a 
M\"obius trasformation $\psi$ of $\hat\C$:
$$
\E X_{g}(\psi(x))X_{g}(\psi(y))
=\E X_{g_\psi}(x)X_{g_\psi}(y) 
$$
where the transformed metric is
\begin{align}\label{gpsi}
g_\psi:=|\psi'|^2g\circ\psi.
\end{align}
Indeed \eqref{covg} may be written as
\begin{align}\label{covg2}
G_g(z,z')=(\int_\C g(v)dv )^{-2}\int_{\C^2}\ln\frac{|z-u||z'-v|}{|z-z'||u-v|}g(u)g(v)dudv.
\end{align}
A change of variables $u=\psi(u')$, $v=\psi(v')$ and invariance of cross ratios under M\"obius maps give the claim.
%$$
%\frac{|\psi(z)-\psi(u')||\psi(z')-\psi(v')|}{|\psi(z)-\psi(z')||\psi(u')-\psi(v')|}
%$$
We may state this as
\begin{align}\label{Xpsi}
X_g\circ\psi\stackrel{law}{=}X_{g_\psi} .%\stackrel{law}{=}X_g-m_{g_\psi}(X_{g})
\end{align}
The random field $X_g$ determines  probabilty measure $\P_g$ on $H^{-1}(\hat \C,g)$ through its  generating function 
\begin{align*}
\E\, e^{i(X_g, f)_g}=\int e^{i(X, f)_g}P_g(dX).
%\label{rffexp}
\end{align*} 
%which is continuous on $\caS(\C)$ of positive type. Conversely such a $W(f)$ (satisfying $W(0)=1$) determines a unique measure $\nu$. 
 We define the {\it Massless Free Field} as the Borel measure $\nu_{MFF}$ on  $H^{-1}(\hat \C,g)$ as the push-forward of the measure $\P_g\times dc$ on $H^{-1}(\hat \C,g)\times\R$ to $H^{-1}(\hat \C,g)$ under the map $(X,c)\to X+c$. Concretely
$$
\int F(X)\nu_{MFF}(dX)=\int_\R (\E\, F(X_g+c))dc
$$
 Note that 
 $\nu_{MFF}$ is {\it not} 
a probability measure: $\int  \nu_{MFF}(dX)=\infty$. Using \eqref{change} we see that this measure  is {\it independent of the chosen metric} in the conformal class of $\hat g$ since  the  random constant
$m_{g'}(X_g)$ can be absorbed to a shift in $c$.

We can now give a tentative definition of the measure in \eqref{funci} by defining
\begin{align}
\nu_{g}(dX)=e^{- \frac{1}{\pi}\int_{\C}(\frac{_Q}{^4}gR_{g} X  + \pi\mu e^{\gamma \phi_g  }\,)dz}\nu_{MFF}(dX) .\label{tenta}
\end{align}
However, now we encounter the problem of renormalization as $ e^{\gamma X_g  }$ is not defined since $X_g$ is not defined pointwise.

\subsection{Multiplicative Chaos} To define $ e^{\gamma X_g  }$ we proceed  by taking a mollified  version of GFF   
 \begin{equation}\label{molli}
X_{g,\epsilon}:=\rho_\epsilon\ast X_g
\end{equation} 
 where $\rho_\epsilon(z)=\epsilon^{-2}\rho(z/\epsilon)$ and $\rho$ is a smooth rotation invariant mollifier.
   We have from \eqref{covg}
 \begin{equation}\label{circlegreen}
%\lim_{\epsilon\to 0}(\
\E X_{g,\epsilon} (z)^2=\ln \epsilon^{-1}+a(\rho)-2c_g(z)+C_g+o(1)
%=-\hf\ln {g}(z)% \ln 2-\hf
\end{equation} 
uniformly on $\C$ where the constant $a(\rho)=%-4\pi^2\int\rho(r)\rho(r')(\ln r\wedge\ln r')drdr'
\int \rho(z)\rho'(z)\ln|z-z'|^{-1}dzdz'$ depends on the regularization function $\rho$. 
Hence
 \begin{equation}\label{circlegreen1}
\epsilon^{\frac{_{\gamma^2}}{^2}} \E e^{\gamma X_{{g},\epsilon}(z)}=\epsilon^{\frac{_{\gamma^2}}{^2}} e^{\frac{\gamma^2}{2}\E X_{g,\epsilon} (z)^2}=(A+o(1))e^{-\gamma^2c_g(z)}
\end{equation} 
for a constant $A$. Hence it is natural to renormalize by defining the  random measure on $\C$
\begin{align}
M_{g,\gamma,\epsilon}(dz):=\epsilon^{\frac{_{\gamma^2}}{^2}}e^{\gamma (X_{{g},\epsilon}(z)+\frac{_{ Q}}{^2}\ln g(z))}%e^{\gamma X_{{g},\epsilon}-\frac{\gamma^2}{2} \E[ X_{{g},\epsilon}^2] }
dz%=(a+o(\epsilon))e^{\gamma X_{{g},\epsilon}-\frac{\gamma^2}{2} \E X_{{g},\epsilon}^2 }g(z)dz
\label{Mdefi}
\end{align}
In particular for the round metric we get
\begin{align}
M_{\hat g,\gamma,\epsilon}(dz)=(a+o(1))e^{\gamma X_{{\hat g},\epsilon}-\frac{\gamma^2}{2} \E X_{{\hat g},\epsilon}^2 }\hat g(z)dz
\label{Mdefiround}
\end{align}
with $a=e^{\hf\gamma^2(\ln 2-\hf)}%C_{\hat g}}
$.
\begin{proposition}\label{basic}
$$
M_{g,\gamma, \epsilon}\to M_{ g,\gamma}  
$$
 weakly in probability as $\epsilon\to 0$. The limit is independent of the mollifier and nonzero if and only if $\gamma<2$. It satisfies
 \vskip 2mm
 \noindent (a)  Let $B_r$ a ball of radius $r$ and $p>0$. Then $\E  M_{ g,\gamma} (B_r)^p<\infty$ if and only if $p<4/\gamma^2$ and then 
 \begin{align}\E  M_{ g,\gamma} (B_r)^p\leq Cr^{\xi(p)}
 \label{Moment1}
\end{align}
where $\xi(p)=\gamma Qp-\hf\gamma^2p^2$.

 \vskip 2mm
 \noindent (b) $\E  M_{ g,\gamma} (B_r)^p<\infty $ for all $p<0$.
 
 \end{proposition}
 The limit is  an example of {\it Gaussian multiplicative chaos} (see \cite{DRSVreview} for a review), a random
{\it multifractal} measure on $\C$. We will use the notation $M_{ g,\gamma}(f)=\int f(z)M_{g,\gamma}(dz)$.
From \eqref{Mdefi} we have $\E M_{ g,\gamma}(1)=\int_\C B(z)g(z)dz$ where $B$ is bounded. Hence  $M_{ g,\gamma}(\C)<\infty$ a.s. We may now define
\eqref{tenta} as
\begin{align}
\nu_g=e^{- \frac{1}{\pi}\int_{\C}\frac{_Q}{^4}R_{g} X gdz - \mu e^{\gamma c  }M_{ g,\gamma}(1)}\nu_{MFF} .\label{tenta}
\end{align}
We will often use the notations 
$$e^{\gamma(X_g+\frac{_Q}{^2}\ln g)}dz:=M_{ g,\gamma}(dz)\ \ \ e^{\gamma\phi_g}dz:=e^{\gamma c}M_{ g,\gamma}(dz)
$$ 
but the reader should be aware that  $M_{ g,\gamma}$ is {\it not} absolutely continuous w.r.t. the Lebesque measure.

The chaos measure has a nice transformation law under conformal maps:
\begin{proposition}\label{mobiusforchaos} Let $\psi$ be a M\"obius map of $\hat C$. Then
\begin{align*}
\int fe^{\gamma X_{ g}}dz=\int f\circ\psi \, e^{\gamma X_{ g}\circ\psi}|\psi'|^{2+\frac{\gamma^2}{2}}dz
%\label{actio}
\end{align*}
%where  $g_\psi$ denotes the transformed metric $ g_\psi=|\psi'|^2g\circ\psi $.

\end{proposition}
\begin{proof}
Making a change of variables  we get
\begin{align}\label{mob}
\int f e^{\gamma X_{ g}}dz&=\lim_{\epsilon\to 0}\int  f \epsilon^{\frac{\gamma^2}{2}}e^{\gamma X_{ g,\epsilon}}dz
=\lim_{\epsilon\to 0}
\int f\circ\psi \epsilon^{\frac{\gamma^2}{2}}e^{\gamma X_{ g,\epsilon}\circ\psi}|\psi'|^2dz. %\\
%&=\lim_{\epsilon\to 0}\int e^{\gamma X_{ g}\circ\psi}|\psi'|^{2+\frac{\gamma^2}{4}}dz=\int e^{\gamma (X_{ g}\circ\psi+\frac{Q}{2}\ln  g_\psi)}dz
\end{align}
Suppose first $\psi$ is the scaling $\psi(z)=\lambda z$. Then 
$$X_{ g,\epsilon}(\psi(z))=\epsilon^{-2}\int \rho(|\lambda z-u|/\epsilon)X_g(u)du=(X_g\circ\psi)_{ \epsilon/|\lambda|}(z)$$
 and the claim follows by setting $\epsilon'=\epsilon/|\lambda|$. For the general case one notes that
\begin{align*}
\lim_{\epsilon\to 0}((\E X_{ g,\epsilon}(\psi(z)) X_{ g,\epsilon}(\psi(u))-\E( X_g\circ\psi)_{ \epsilon/|\psi'(z)|}(z)(X_g\circ\psi)_{ \epsilon/|\psi'(u)|}(u))=0
\end{align*}
uniformly on compacts in $\C\setminus \{\psi^{-1}(\infty)\}$ and invokes uniqueness of the chaos measure under such condition. % \cite{shamov}.
\end{proof}
Note that by \eqref{Xpsi}  we get in partcular
\begin{align}\label{invar}
\int_\C e^{\gamma\phi_g}dz\stackrel{law}=\int_\C e^{\gamma\phi_{g_\psi}}dz.
\end{align}

\subsection{Weyl and M\"obius invariance} We saw that $X$ is metric independent under $\nu_{MFF}$. For the { Liouville field} we have (compare with \eqref{weyl0})

\begin{proposition}\label{Weyl}
Let $F\in L^1(\nu_g)$ and %$g,g'$ conformally equivalent with 
$g'=e^{\varphi}g$. Then%its law under $\nu_L$ is independent of the metric $g$ modulo a multiplicative constant.
\begin{align*}
\int F(\phi_{g'})d\nu_{g'}=e^{\frac{c_L-1}{24\pi}\int (|\partial_z \varphi|^2+ \hf gR_{{g}}\varphi )dz}\int F(\phi_g)d\nu_{g}
%\label{actio}
\end{align*}
where $c_L=1+6Q^2$.
\end{proposition}
\begin{proof}
%e^{\gamma c}M_{g,\gamma,\epsilon}(dz)=\epsilon^{\frac{_{\gamma^2}}{^2}}e^{\gamma\phi_{g,\epsilon}}dz%:=e^{\gamma\phi_{g}}dz.$. Thus b
By  metric independence of $X$ we replace $c+X_{g'}$ by $c+X_{g}$ so that
%$$
%e^{\gamma c}M_g(dz)=\lim_{\epsilon\to 0}\epsilon^{\frac{_{\gamma^2}}{^2}}e^{\gamma\phi_{g,\epsilon}}dz:=e^{\gamma\phi_{g}}dz.
%$$
\begin{align*}
\int F(\phi_{g'})d\nu_{g'}=\int F(\phi_{g}+\frac{_Q}{^2} \varphi)e^{- \frac{1}{\pi}\int( \frac{_Q}{^4}R_{g'} g'(c+X_g)+\pi\mu e^{\gamma(\phi_g+\frac{_Q}{^2} \varphi)})dz}% e^{\gamma c} M_{g,\gamma}( e^{\hf Q%\frac{_Q}{^2} \varphi}))}
d\nu_{MFF} .
%\label{actio}
\end{align*}
Use $ R_{g'} g'=-\Delta \ln g'=R_{g} g-\Delta\varphi$ and  Gauss-Bonnet theorem $\int R_{g'} g'=8\pi=\int R_{g} g$ to get
\begin{align*}
\int R_{g'} g'(c+X_g)dz=\int R_{g} g(c+X_g)dz-\int \Delta\varphi X_gdz.
%\label{actio}
\end{align*}
Hence
\begin{align*}
\int F(\phi_{g'})d\nu_{g'}=\int F(\phi_{g}+\frac{_Q}{^2} \varphi)e^{- \frac{1}{\pi}\int( \frac{_Q}{^4}R_{g} g(c+X_g)+\pi\mu e^{\gamma(\phi_g+\frac{_Q}{^2} \varphi)})dz}
e^{ \frac{Q}{4\pi}%\int \Delta\phi X_gdz
(X_g,\Delta_g\varphi)_g} d\nu_{MFF} .
%\label{actio}
\end{align*}
The result then follows by  a shift in the Gaussian integral $X_g\to X_g-\frac{_Q}{^2}( \varphi-m_g\varphi)$ and $c$ to $c-m_g\varphi$. Indeed,  by %an application of the
Girsanov theorem %. Indeed, recall the covariance \eqref{covg}. By Girsanov 
the law of $X_g- \frac{Q}{4\pi}G_g\Delta_g\varphi$ under the measure
$$e^{ \frac{Q}{4\pi}(X_g,\Delta_g\varphi)_g- \frac{Q^2}{32\pi^2}
\E((X_g,\Delta_g\varphi)_g)^2}\P_g
$$ 
equals the law of  $X_g$ under $\P_g$. We use here the notation $(G_gf)(z):=\int G_g(z,w)f(w)g(w)dw$. Since from \eqref{deltag} 
$\Delta_x G_g(x,y)=-2\pi\delta(x-y)+2\pi g(x)
$ we have
$G_g\Delta_g\varphi=-2\pi(\varphi-m_g\varphi)$
 so that
$$
X_g- \frac{Q}{4\pi}G_g\Delta_g\varphi=X_g+\frac{_Q}{^2}( \varphi-m_g\varphi)
$$
and we end up with
\begin{align*}
\int F(\phi_{g'})d\nu_{g'}=%e^{\frac{c_L-1}{96\pi}\int |\partial \varphi|^2\,dz+ \int 2R_{{g}}\varphi \,gdz}
e^{A(\varphi,g)}\int F(\phi_{g})d\nu_{g}
%\label{actio}
\end{align*}
where
$$
A(\varphi,g)= \frac{Q^2}{32\pi^2}\E(X_g,\Delta_g\varphi)_g^2+ \frac{Q^2}{8\pi}\int R_{g} g\varphi
$$
The claim follows since $\E(X_g,\Delta_g\varphi)_g^2=(\varphi,\Delta_g G_g\Delta_g \varphi)_g =-2\pi(\varphi,\Delta_g \varphi)_g =8\pi\int|\partial_z\varphi|^2dz$.\end{proof}
 
\begin{remark}
\noindent The multiplicative factor is called the {\it Weyl anomaly} in physics literature and $c_L$ is the  {\it central charge} of Liouville theory. Usually the Weyl transformation law in CFT has the factor $c_L$ and not  $c_L-1$. The reason for this discrepancy is that we have used the {\it normalized} law $\P_g$ for the GFF instead of the unnormalized one as in \eqref{gffexp}. To get the unnormalized law one needs to multiply by the partition function $Z_g$ of the GFF $X_g$ which formally is given by
$$
\det(-\Delta^\perp_g)^{-\hf}=\prod_{n>0}\lambda_n^{-\hf}
$$ While this is not defined %infinite ($\lambda_n\to\infty$ as $n\to\infty$) 
its variation under Weyl transformation can be defined and the upshot is that $c_L-1$ gets replaced by $c_L$.
\end{remark}
As a consequence of the Proposition we get M\"obius transformation rule:
\begin{corollary}\label{mobius} Let $\psi$ be a M\"obius map of $\hat \C$. Then
\begin{align*}
\int F(\phi_{ g}\circ\psi)d\nu_{ g}=\int F(\phi_{ g}-Q\ln|\psi'|)d\nu_{ g}
%\label{actio}
\end{align*}
\end{corollary}
\begin{proof} From \eqref{mob} we get%Making a change of variables in the  Liouville term we get
\begin{align*}
\int e^{\gamma\phi_{ g}}dz=%\int (e^{\gamma\phi_{ g}}\circ\psi)|\psi'|^2dz=\int e^{\gamma\phi_{ g}\circ\psi}|\psi'|^{2+\frac{\gamma^2}{4}}dz=
\int e^{\gamma(c+X_{ g}\circ\psi+\frac{Q}{2}\ln  g_\psi)}dz
\end{align*}
%where  $g_\psi$ denotes the transformed metric
%$g_\psi=|\psi'|^2g\circ\psi$.
For  the curvature term we get
\begin{align*}
\int gR_gX_g=-\int\Delta\ln g X_g=-\int\Delta(\ln g\circ\psi) X_g\circ\psi=-\int\Delta(\ln g_\psi) X_g\circ\psi=\int g_\psi R_{g_\psi}X_g\circ\psi
\end{align*}
where in the third step we used that $\ln|\psi'|$ is harmonic.
Recalling that
$
X_g\circ\psi\stackrel{law}{=}X_{g_\psi} %\stackrel{law}{=}X_g-m_{g_\psi}(X_{g})
$
and combining with   Proposition \ref{Weyl} we obtain
\begin{align*}
\int F(\phi_{g}\circ\psi)d\nu_{ g}=\int F(\phi_{ g_\psi}-Q\ln|\psi'|)d\nu_{g_\psi}=e^{A(\varphi,g)}\int F(\phi_{ g}-Q\ln|\psi'|)d\nu_{g}
%\label{actio}
\end{align*}
where $\varphi=\ln|\psi|^2+\ln g\circ\psi-\ln g$. Using the fact that $\ln|\psi|^2$ is harmonic we get
$A(\varphi,g)=A(\ln g\circ\psi-\ln g,g)$ and some algebra shows this vanishes.
 \end{proof}

\subsection{Vertex operators} Since the  M\"obius group is non-compact the Corollary makes one suspect that the measure $\nu_g$ does not have a  finite mass. Let us consider its Laplace transform. By Proposition \ref{Weyl} we may work with the round metric $\hat g$ where $R_{\hat g}=2$. Then  $\frac{1}{4\pi}\int R_{\hat{g}} (c+X_{\hat g})g dz=2c$  since $\int X_{\hat g}\hat g=0$. We get
\begin{align*}
\int e^{(X,f)_{\hat g} }d\nu_{\hat g}=\int\,e^{((1,f)_{\hat g} - 2Q)c}  \E_{\hat g}\, e^{ (X_{\hat g},f)_{\hat g}  } e^{-  \mu e^{\gamma c}\int_\C M_{\hat g,\gamma}{(dz)}}%\int e^{\gamma X_{\hat{g}}}{\hat g}^{\frac{_{\ga Q}}{2}} \la}
dc
%\label{actio}
\end{align*}
Since  $\int M_{\hat g,\gamma}{(dz)}<\infty$ a.s. the integral converges if and only if $(1,f)_{\hat g} >2Q$. In particular taking $f=0$ we see that the total mass of $\nu_{\hat g}$ is infinite. We can do the $c$-integral to get 
\begin{align*}%\label{gener}
\int e^{( X,f)_{\hat g}   }\nu_{\hat g}(dX)&=
\gamma^{-1}{\mu^{-s_f}}\Gamma(s_f)  \E_{\hat g}\,\big( e^{(X_{\hat g},f)_{\hat g}  } (\int_\C M_{ \hat g,\gamma}(dz))^{-s_f}\big)%\\&=\gamma^{-1}{\mu^{-s_f}}\Gamma(s_f)e^{\hf (f,G_{\hat g} f)_{\hat g} }  \E_{\hat g}\, M_{ \hat g,\gamma}(e^{\gamma G_{\hat g}f})^{-s_f}
\end{align*} 
where we denoted $s_f=\gamma^{-1}((1,f)_{\hat g}- 2Q)$. We may further simplify this my a shift in the gaussian integral  i.e. by a use of the  Girsanov theorem:
\begin{align}\label{gener}
\int e^{( X,f)_{\hat g}   }\nu_{\hat g}(dX)&=
%\gamma^{-1}{\mu^{-s_f}}\Gamma(s_f)  \E_{\hat g}\, e^{(X_{\hat g},f)_{\hat g}  } (\int_\C M_{ \hat g,\gamma}(dz))^{-s_f}\\&=
\gamma^{-1}{\mu^{-s_f}}\Gamma(s_f)e^{\hf (f,G_{\hat g} f)_{\hat g} }  \E_{\hat g}\, \big(\int e^{\gamma (G_{\hat g}f)(z)}M_{ \hat g,\gamma}%(e^{\gamma G_{\hat g}f}
(dz)\big)^{-s_f}
\end{align} 
%Recall that we are looking for a measure with three points on $\hat \C$ fixed. 
Note how the Laplace transform of the Gaussian measure is modified in the Liouville theory.

We define (regularized) {vertex operators}
$$
%e^{\al X(z)}
V_{\alpha,\epsilon}(z):=\epsilon^{\frac{_{\alpha^2}}{^2}}e^{\alpha\phi_{g,\epsilon}(z)}
$$
and consider their correlation function 
$$
%Z({\bf z},{\alpha},g)
\langle \prod_{i=1}^nV_{\alpha_i}(z_i)\rangle_g:=\lim_{\epsilon\to 0}\int  \prod_{i=1}^nV_{\alpha_i,\epsilon}(z_i)d\nu_{{g}}.$$
Plugging in \eqref{gener} $f=\hat g^{-1}\sum_i\alpha_i\rho_\epsilon\ast \delta_{z_i}$ and recalling
 \eqref{Mdefiround} we get
 \begin{align*}%\label{gener1}
\langle \prod_{i=1}^nV_{\alpha_i}(z_i)\rangle_{\hat g}=C(\alpha)\gamma^{-1}{\mu^{-s}}\Gamma(s)  e^{\sum_{i<j}G_{\hat g}(z_i,z_j)\alpha_i\alpha_j}e^{\sum_i\alpha_i\frac{_Q}{^2}\ln\hat g(z_i)}  \E_{\hat g}\, \big(\int e^{\gamma\sum_i\alpha_iG_{\hat g}(z,z_i)}M_{ \hat g,\gamma}(dz)\big)^{-s}
\end{align*} 
with 
$s=\gamma^{-1}(\sum_i\alpha_i-2Q)$ and we need  the condition
 \begin{align}\label{Seib1}
\sum_i\alpha_i>2Q
\end{align} 
for convergence of the $c$-integral. 
%\begin{remark}
%We glossed over here an easy argument to show 
%$$
%\lim_{\epsilon\to 0}\E_{\hat g}\, M_{ \hat g,\gamma}(e^{\gamma\sum_i\alpha_iG_{\hat g}\hat g^{-1}\rho_\epsilon\ast{\delta_{z_i}}})^{-s}=\E_{\hat g}\, M_{ \hat g,\gamma}(e^{\gamma\sum_i\alpha_iG_{\hat g}(\cdot,z_i)})^{-s}
%$$
%see \cite{DKRV}.
%\end{remark}
Using the expression for the Green function \eqref{covg1} ,\eqref{covg2}
%The result is after some calculation (\cite{DKRV}, \cite{KRV})
we arrive at
\begin{equation*}
% \langle    \prod_{l} V_{\alpha_l}(z_l)  \rangle  
%Z({\bf z},\alpha,g) 
\langle \prod_{i=1}^nV_{\alpha_i}(z_i)\rangle_{\hat g}=C'(\alpha) \prod_{j < k} \frac{1}{|z_j-z_k|^{\alpha_j \alpha_k}} \mu^{-s} \gamma^{-1}\Gamma(s)\E\,  \big(\int F(z)M_{ \hat g,\gamma}(dz)\big)^{-s}  %\label{Z1}
\end{equation*}
where
\begin{equation*}
F(z)=   \prod_i \frac{1}% { \hat g}(z)^{ - \frac{\gamma}{4} \alpha_i }
{|z-z_i|^{\gamma \alpha_i}}  { \hat g}(z)^{ - \frac{\gamma}{4} \alpha_i }
 .
\end{equation*}
Note that the expectation is finite due to Proposition \ref{basic} since% picking a ball $B_r$ in the complement of the points $z_i$ we have
$$
M_{ \hat g,\gamma}(F)^{-s} \leq M_{ \hat g,\gamma}(F1_{B_r})^{-s}\leq (\inf_{z\in B_r}F(z))^{-s} M_{ \hat g,\gamma}(B_r)^{-s}.
$$
However, the expectation may be zero due to blowup of the integral. This is the content of
\begin{proposition}\label{vetrexp}
$0<\langle \prod_{i=1}^nV_{\alpha_i}(z_i)\rangle_g<\infty$ if and only if {$\sum\alpha_i>2Q$} and  $\alpha_i<Q$.
\end{proposition}
These bounds for $\alpha_i$ are called {\it Seiberg bounds}. Note that they imply that we need  {\it at least three} vertex operators to have a finite correlation function.

\noindent {\it Scetch of proof.} Let $Z=M_{ \hat g,\gamma}(F)$ and write $Z=\sum_{i=0}^{n}Z_i$ where in $Z_i$ the integration is over a small ball around $z_i$ if $i>0$ and in the complement of all the balls if $i=0$. Then, for $0<p<1$
\begin{align*}%\label{gener1}
\E Z^{-s}\geq r^{-s}\P(Z<r)= r^{-s}(1-\P(Z>r))\geq  r^{-s}(1-r^{-p}\E Z^p)\geq  r^{-s}(1-r^{-p}\sum_i\E Z_i^p)
\end{align*} 
where in the last step we used subadditivity. By Proposition \ref{basic} $\E Z_0^p<\infty$ for $p<4/\gamma^2$ since $F<C$ on the support. Thus $\E Z^{-s}>0$ follows from $\E Z_i^p<\infty$ for some $0<p<1$ and all $i$. On the other hand,
\begin{align*}%\label{gener1}
\E Z^{-s}\leq \E Z_i^{-s}
\end{align*} 
so $\E Z^{-s}=0$ follows from $ \E Z_i^{-s}=0$ for some $i$.

For the first case, we use Kahane convexity (see \cite{review}). % Proposition \ref{basic}: i
In $B_i$ we can bound
$$
\ln |z-u|^{-1}_+\leq G_{\hat g}(z,u)+A
$$ 
so that comparing chaos with field $X_+$ to $X_{\hat g}+n$ where $n$ is normal with variance $A$ Kahane convexity gives us   for $0<p<1$  
$$
\E Z_i^p\leq C \E( \int_\D|z|^{-\gamma\alpha_i}dM_+)^p
$$
Write $\D=\cup_{n=0}^\infty A_n$ where $A_n$ is annulus with radi $2^n$ and $2^{n+1}$. 
The field $X_+$ satisfies 
$$X_+(2^{-n} \cdot)\stackrel{law} =X_+( \cdot)+x_n$$ where the summands are independent and $x_n$ is normal with variance $ \ln 2^n$. Hence
$$
 I_n:=\int_{A_n}|z|^{-\gamma\alpha_i}dM_+\stackrel{law}=e^{\gamma x_n-\hf\gamma^2 \ln 2^n} 2^{(\gamma\alpha_i-2)n}I_0=e^{\gamma x_n} 2^{\gamma(\alpha_i-Q)n}I_0
$$
and 
$$
\E Z_i^p\leq C \sum_n\E I_n^p\leq C \sum_n 2^{(\gamma\alpha_i-\gamma Q)np}\E I_0^p\,\E e^{\gamma px_n}\leq C \sum_n 2^{(\gamma\alpha_i-\gamma Q)np+\hf\gamma^2p^2n}
$$
which converges if $p<\frac{2(Q-\alpha_i)}{\gamma}$.

For the second claim by M\"obius invariance  it suffices to suppose $z_i=0$ and $B_i=\D$.
We use the following "radial" decomposition of the GFF
 (see \cite{cf:DuSh,DKRV2}). Let 
 \begin{align*}%\label{gener1}
X_{\hat g,r}:=\frac{1}{2\pi}\int X_{\hat g}(re^{i\theta})d\theta.
\end{align*} 
Then 
\begin{align*}%\label{gener1}
X_{\hat g}(z)= X_{\hat g,|z|}+Y(z)
\end{align*} 
where the fields on the RHS are independent and $Y$ has the covariance
\begin{equation*}
\E   Y(re^{i\theta}) Y(r'e^{i\theta'})   = \ln \frac{ r \vee r'}{|re^{i \theta}-r'e^{i \theta'}|}
\end{equation*}
and the process 
$$ B_t:=  X_{\hat{g}, e^{-t}}(0)%-X_{\hat{g},1}(0)  
$$
is a Brownian motion starting at $B_0= X_{\hat{g},1}(0)$, an independent gaussian variable of variance $\caO(1)$. This leads to the following expression for the chaos
$$ \int_\D\frac{1}{|x|^{\gamma \alpha}}\,M_\gamma(dx) =\int_0^\infty \int_0^{2\pi}e^{\gamma B_t-(Q-\alpha)t%X_{\hat{g},e^{-r}}(0)
}\,\mu_{Y}(dt,d\theta).$$
where $$\mu_{Y}(dt,d\theta)=const. e^{\gamma Y(e^{-r+i\theta})-\frac{\gamma^2}{2}\E[Y(e^{-r+i\theta})^2]}\hat{g}(e^{-r})\,d\theta dr $$  is a chaos measure independent of the  process $B_t$. %(X_{\hat{g},e^{-r}}(0))_{r\geq 0}$.
%Furthermore, for all $q\in]-\infty;\frac{4}{\gamma^2}[$, we have
%\begin{equation}\label{momentbound}
%\sup_{a>0}\E\Big[ \Big(\int_a^{a+1}\int_0^{2\pi}e^{\gamma(X_{\hat{g},e^{-r}}(0)-X_{\hat{g},e^{-a-1}}(0) )} \mu_Y(dr,d\theta)\Big)^{q}\Big]<+\infty.
%\end{equation}
Note that the drift term vanishes as $\alpha\to Q$. Let
$$
Z_t:=\int_0^t \int_0^{2\pi}e^{\gamma B_s%X_{\hat{g},e^{-r}}(0)
}\,\mu_{Y}(ds,d\theta).
$$
Recall that $\P(\sup_{s<t}B_s<k)\leq kt^{-\hf}$. This leads us to expect that $\E Z_t^{-s}=\caO(t^{-\hf})$. Indeed, this is true:
$$\lim_{t\to\infty}  t^\hf\E Z_t^{-s}
$$ 
exits and is nonzero. Therefore the correct normalization of the vertex operators for $\alpha=Q$ is
$$V_{Q,\epsilon}=(\ln \epsilon^{-1})^\hf\epsilon^{\frac{Q^2}{2}}e^{Q\phi_{\hat g\epsilon}}.$$ Then
Proposition \ref{vetrexp} holds also for $\alpha_i\leq Q$.

\subsection{KPZ Conjecture for Measure and Correlations}\label{lme}

%Given $\alpha_i$ satisfying the Seiberg bounds we can define the probability measure
Let us now return the scaling limit of random triangulations. Define the probability measure
$$
d\P_{{\bf z},\gamma}:=%Z({\bf z},\alpha
\langle \prod_{i=1}^3V_{\gamma}(z_i)\rangle_{\hat g}^{-1}\prod_{i=1}^3V_{\gamma}(z_i)d\nu_{ \hat g}$$
%$$
%\phi |_{\P_{\alpha_i, z_i}}\stackrel{law}{=}(\phi\circ\psi+Q\log|\psi'|)|_{\P_{\alpha_i, \psi(z)_i}}
%$$
%The law of the Liouville field $\phi_{\hat g}$ under $\P_{{\bf z},\gamma}$ is independent on the metric $g$ since the Weyl anomaly cancels. In particular the law of the {\it Liouville measure
We may then state the KPZ conjecture for the random measure $\rho_{\mu,\gamma}$ obtained from scaling limit of triangulations. We conjecture that the law of   $\rho_{\mu,\gamma}$ equals the law of the measure $e^{\gamma\phi_{\hat g}}dz:=e^{\gamma c}M_{ g,\gamma}(dz)$ under 
$\P_{{\bf z},\gamma}$. % and depends only on the conformal structure i.e. the $. 
Let us check that the law for the total volume $A=\int_\C e^{\gamma\phi_{\hat g}}dz$ matches. By a simple change of variables in the $c$-integration $e^{\gamma c}M_{g,\gamma}(\C)=A$ we obtain
$$
 \E F(A) =\frac{\mu^s}{\Gamma (s)}\int_0^{\infty}F(y)y^{s-1}e^{-\mu y}\,dy.$$
%$$s= ({\sum_i \alpha_i-2Q})/\gamma$$
where $s=(3\gamma-2Q)/\gamma=2-4/\gamma^2$ i.e. under $\P_{{\bf z}, \gamma}$ the law of  $A$ is  $\Gamma(2-4/\gamma^2,\mu ) $. This  agrees with the result in random surfaces.
 %Planar Maps. %
 Note that the conformal weight $\Delta_\gamma=1$ (see next Section) so the vertex operator $e^{\gamma \phi_g}$ transforms under conformal maps as a density.
 
 For the Ising model random field \eqref{ising} the KPZ conjecture says that its correlation functions converge
 \begin{align*}
\lim_{\epsilon\to 0}\E\Phi^{(\epsilon)}(u_1)\dots \Phi^{(\epsilon)}(u_n)=\E\sigma(u_1)\dots \sigma(u_n)\E_{{\bf z}, \gamma}V_\alpha(u_1)\dots V_\alpha(u_n)
\end{align*}
where $\E\sigma(u_1)\dots \sigma(u_n)$ are the correlation functions of the Ising model in the scaling limit on $\hat\C$ and  $\alpha$ is determined from the requirement $\frac{1}{16}+\Delta_{\alpha}=1$ which means that $\sigma(z)e^{\alpha\phi_g(z)}$  transforms under conformal maps as a density.
\subsection{Conformal Ward Identities}
So far we have motivated the Liouville model through its conjectural relationship  to scaling limits of random triangulations. However, the Liouville model  is also an interesting {\it Conformal Field Theory} by itself. This way of looking we view the vertex operators as (Euclidean) quantum fields.% and the $Z({\bf z},\alpha)$ their "correlation functions". Due to non-normalizability of the measure $\nu_g$ we can not view these as expectations there is a meaningful quantum theory behind them. 

First, using the M\"obius invariance (Corollary \ref{mobius})  of $\nu_g$  and taking care with the transformation of the $\epsilon$ in the vertex operator one gets 
\begin{align*}
\langle \prod_{i=1}^nV_{\alpha_i}(\psi(z_i))\rangle_g&=\lim_{\epsilon\to 0}\langle \prod_{i=1}^n\epsilon^{\frac{\alpha_i^2}{2}}e^{\alpha_i\phi_{g,\epsilon}(\psi(z_i))}\rangle_g=\lim_{\epsilon\to 0}\langle \prod_{i=1}^n|\epsilon\psi'(z_i)|^{\frac{\alpha_i^2}{2}}e^{\alpha_i(\phi_{g}\circ\psi)_\epsilon((z_i)}\rangle_g
\\
&=\prod_i|\psi'(z_i)|^{-2\Delta_{\alpha_i}}\langle \prod_{i=1}^nV_{\alpha_i}(z_i)\rangle_g
\end{align*}
where $\Delta_{\alpha}=\frac{\alpha}{2}(Q-\frac{\alpha}{2})$. In CFT parlance, $V_\alpha$ is a {\it primary field} with conformal weight $\Delta_{\alpha}$.

Second, the Liouville model has also {\it local conformal symmetry}. In CFT  this derives from the  {\it energy-momentum tensor}  which encodes the variations of the theory with respect to the background metric. In classical field theory this is defined as follows. Let $S(g,X)$ be an action functional where $g=%\sum_{\alpha,\beta}
g_{\alpha\beta}dx^\alpha\otimes dx^{\beta}$ (we use summation convention of repeated indices) is a smooth Riemannian metric. In Liouville case
$$
S(g,X)=%\frac{1}{4\pi}
\int (g^{\alpha\beta}\partial_\alpha X\partial_\beta X+QR_gX+\mu e^{\gamma X})\sqrt{\det g}dx
$$ 
where $g^{\alpha\beta}$ is the inverse matrix  $g^{\alpha\beta}g_{\beta\gamma}=\delta^\alpha_\gamma$. Then the EM tensor $T_{\alpha\beta}(x)$ is defined by
$$
\partial_\epsilon|_0S(g_\epsilon,X)=\int T_{\alpha\beta}(x)f^{\alpha\beta}(x)\sqrt{\det g}dx
$$
where $g_\epsilon^{\alpha\beta}=g^{\alpha\beta}+\epsilon f^{\alpha\beta}$. For Liouville model one finds that the only interesting component of $T$ in complex coordinates is $T_{zz}:=T(z)$ which is classically analytic $\partial_{\bar z}T=0$ if $X$ satisfies the Euler-Lagrange equations. In quantum theory one defines in the same way
\begin{equation}\label{defgen}  
  %\frac{1}{\delta g^{zz}} 
 \frac{_d}{^{d\epsilon}}\mid_{\epsilon=0} \langle   \prod_l V_{\alpha_l}(z_l)   \rangle_{g_\epsilon}:=  \frac{1}{4\pi}   \int f^{\alpha\beta}(z)\langle T_{\alpha\beta}(z) \prod_l V_{\alpha_l}(z_l)    \rangle_g  dz.
 \end{equation}
for $f$  a smooth function with support in $\C\setminus \cup_i z_i$. 

%More specifically, one may define the the correlation functions in a smooth Riemannian metric near our  ${g}$ and consider the one parameter family $g^{-1}_\epsilon=g^{-1}+\epsilon f \partial_z\otimes \partial_z
%$ where 
% Then
 % if $F$ is some smooth function on $H^{-1}$ and $V_{\alpha_l}$ vertex operators, 
% (a component of) the stress tensor $T(z)$ is defined by the following formula in the physics literature (see \cite{gaw})
 %\begin{equation}\label{defgen}  
  %\frac{1}{\delta g^{zz}} 
 %\frac{_d}{^{d\epsilon}}\mid_{\epsilon=0} \langle   \prod_l V_{\alpha_l}(z_l)   \rangle_{g_\epsilon}:=     \int f(z)\langle T(z) \prod_l V_{\alpha_l}(z_l)    \rangle_g  dz.
 %\end{equation}
 A simple formal computation then yields the following  heuristic formula %\textcolor{red}{(awkward formulation as they are 3 components for the stress)}
 \begin{equation}\label{defforus}
 T(z)=  Q \partial_{z}^2 \phi(z)- (( \partial_{z}\phi(z))^2-\E( \partial_{z}X_g(z))^2)
 \end{equation}
 where $\phi$ is the { Liouville field}.% and $:-:$ is Wick ordering \eqref{wick}.

  $T(z)$ encodes { local conformal symmetries} through  the {\it Conformal Ward Identities}. The first Ward identity says the correlation function is meromorphic in the argument of $T(z)$ with prescribed singularities: 
\begin{equation}
  \langle T(z) \prod_l V_{\alpha_l}(z_l)   \rangle= \sum_{k} \frac{\Delta_{\alpha_k} }{(z-z_k)^2} \langle  \prod_l V_{\alpha_l}(z_l)   \rangle   -\sum_{k} \frac{1}{z-z_k} \partial_{z_k}\langle  \prod_l V_{\alpha_l}(z_l)   \rangle  \quad\label{wardid1}
 \end {equation}
 and %Note in particular that the $T$ insertion  is holomorphic. 
 the second identity controls the singularity when two $T$-insertions come close  \begin{align}  \nonumber
    \langle T(z)T(z') \prod_l V_{\alpha_l}(z_l)   \rangle
&=\frac{\hf c_{\mathrm{L}} }{(z-z')^4}  \langle T(z')T(z) \prod_l V_{\alpha_l}(z_l)   \rangle +\frac{2 }{(z-z')^2} \langle T(z') \prod_l V_{\alpha_l}(z_l)   \rangle\\& +\frac{1 }{z-z'} \partial_{z'} \langle T(z') \prod_l V_{\alpha_l}(z_l)   \rangle+\dots\label{wardid2}% \textbf{Ward identity}.
 \end {align}
where the dots refer to terms that are bounded as $z\to z'$. To prove these identities we need to define what we mean by the LHS. Let $\phi_\epsilon$ be a regularization of the Liouville field. Set
 \begin{equation}\label{defforus}
 T_\epsilon(z)=  Q \partial_{z}^2 \phi_\epsilon(z)- (( \partial_{z}\phi_\epsilon(z))^2-\E( \partial_{z}X_{\hat g,\epsilon}(z))^2)
 \end{equation}
 and define $ \langle T(z) \prod_l V_{\alpha_l}(z_l)   \rangle$ as the limit of $ \langle T_\epsilon(z) \prod_l V_{\alpha_l}(z_l)   \rangle$ and similarly for the two $T$ insertions.
Let us see how the first Ward identyty follows by formal calculation before commenting on the mathematical problems in actually making it rigorous.

The basic formula is the following  identity:
%\begin{proposition}\label{ippid} 
\begin{align} 
 \langle \partial_z\phi(z)  \prod_k  V_{\alpha_k  }(z_k)\rangle &=-\hf \sum_i\alpha_i \frac{1}{z-z_i}%G_\epsilon( \z) %
  \langle  \prod_k  V_{\alpha_k  }(z_k)\rangle
 +\hf\mu\gamma\int  \frac{1}{z-y}%G(y; \z)%
 \langle  V_{\gamma }(y)\prod_k  V_{\alpha_k }(z_k)\rangle
 dy\label{ipp1}
\end{align}
%\end{proposition}
%\begin{proof} 
To prove this first note that  by integration by parts in the Gaussian measure:
\begin{align} \nonumber
 \langle X_{\hat g}(z)  \prod_k  V_{\alpha_k }(z_k)\rangle =& \sum_i\alpha_iG_{\hat g}(z, z_i)  \langle  \prod_k  V_{\alpha_k  }(z_k)\rangle -\mu\gamma\int G_{\hat g}(z, y) \langle  V_{\gamma }(y)\prod_k  V_{\alpha_k  }(z_k)\rangle dy\label{ipp}.
\end{align}
Recalling the definition of the Liouville field \eqref{Liouville field} we then get
\begin{align*} 
 \langle \partial_z\phi(z)  \prod_k  V_{\alpha_k }(z_k)\rangle &=-\hf \sum_i\alpha_i \frac{1}{z-z_i}%G_\epsilon( \z)% 
 \langle  \prod_k  V_{\alpha_k }(z_k)\rangle +\hf\mu\gamma\int  \frac{1}{z-y}%G(y; \z)%
 \langle  V_{\gamma }(y)\prod_k  V_{\alpha_k }(z_k)\rangle
 dy
 %-\mu\gamma\int(C_\epsilon\ast f)(y)G_\epsilon(y; \z)%\langle  V_{\gamma,\epsilon }(y)\prod_k  V_{\alpha_k ,\epsilon }(z_k)\rangle_{\epsilon}dy
 \\&
 -\frac{_1}{^4}\partial_z\ln\hat g(z)\big((2Q-\sum\alpha_i)  \prod_k  V_{\alpha_k }(z_k)\rangle%G_\epsilon( \z)
 +\mu\gamma\int  \langle  V_{\gamma }(y)\prod_k  V_{\alpha_k }(z_k)\rangle%G_\epsilon(y; \z)
 dy\big) %\label{ipp1}.
\end{align*}
The metric dependent term actually vanishes due to the following  identity
\begin{lemma}\label{kpzid}  (KPZ-identity)
\begin{equation} 
\mu\gamma\int  \langle V_{\gamma}(y)  \prod_k  V_{\alpha_k  }(z_k)\rangle 
= (\sum_i \alpha_i-2Q) \,   \langle  \prod_k  V_{\alpha_k }(z_k)\rangle  .
\label{kpzid}
\end{equation}
\end{lemma}

\proof By a simple change of variables $\gamma^{-1}\ln \mu+c=c'$, we get
\begin{align*} 
  \langle  \prod_k  V_{\alpha_k  }(z_k)\rangle  = &\int_{\R}   e^{-2Q c}  \:   \E[   \prod_k  V_{\alpha_k}(z_k)        e^{- \mu   \int_{\C}e^{\gamma\phi}\,dz    }]   \: dc\\
  = &\mu^{-\frac{\sum_i\alpha_i-2Q}{\gamma}} \int_{\R}   e^{-2Q c'}  \:   \E[   \prod_k  V_{\alpha_k}(z_k)        e^{-   \int_{\C}e^{\gamma\phi}\,dz   }]   \: dc'.
% \\=& \mu^{-\frac{\sum_i\alpha_i-2Q}{\gamma}}  \langle  \prod_k  V_{\alpha_k,\epsilon  }(z_k)\rangle_{\hat{g},\epsilon,\mu=1} 
\end{align*}
The identity follows by differentiating in $\mu$. \qed

Using the integration by parts formula \eqref{ipp1}  we get
 \begin{align*}
 \langle\partial_{z}^2\phi(z) )\prod_l V_{\alpha_l}(z_l)   \rangle&=
\hf\sum_{i}\alpha_i  \frac{1}{(z-z_i)^2} \langle\prod_k  V_{\alpha_k }(z_k)\rangle%G(\z)
-\hf\mu\gamma\int
   \frac{1}{(z-y)^2} %G(y;\z)
   \langle   V_\gamma(y) \prod_kV_{\alpha_k }(z_k)\rangle dy
 \end{align*}
and
 \begin{align*}
  \langle( (\partial_z \phi(z))^2-& \E (\partial_z X(z))^2  ) \prod_l V_{\alpha_l}(z_l)   \rangle
=
\frac{_1}{^4} \sum_{j,k} \frac{ \alpha_j \alpha_k}{(z-z_k)(z-z_j)} \langle\prod_k  V_{\alpha_k }(z_k)\rangle%G(\z)
 \\& -\hf \mu \gamma     \sum_{k} \alpha_k  \frac{1}{(z-z_k)} \int \frac{1}{z-y} \langle   V_\gamma(y) \prod_kV_{\alpha_k }(z_k)\rangle%G(y;\z)
 dy -\frac{_1}{^4} \mu \gamma^2     \int   \frac{1}{(z-y)^2} \langle   V_\gamma(y) \prod_kV_{\alpha_k }(z_k)\rangle
 %G(y;\z)
 dy\\&+\frac{_1}{^4} \mu^2 \gamma^2 \int\frac{1}{z-y} \frac{1}{z-x} \langle   V_\gamma(y) V_\gamma(x) \prod_kV_{\alpha_k }(z_k)\rangle%G(y,x;\z)
 dy dx
 \end{align*}
 Combining we get
\begin{align*}
\langle T(z)\prod V_{\alpha_i}(\z)\rangle&=
(\frac{_Q}{^2}\sum_{i}\alpha_i  \frac{1}{(z-z_i)^2} -\frac{_1}{^4} \sum_{j,k} \frac{ \alpha_j \alpha_k}{(z-z_k)(z-z_j)}) \langle\prod_k  V_{\alpha_k }(z_k)\rangle%G(\z)
\nonumber \\& +\hf\mu  \gamma   \sum_{k} \alpha_k  \frac{1}{z-z_k} \int \frac{1}{z-y}%G(y;\z)
\langle   V_\gamma(y)  \prod_kV_{\alpha_k }(z_k)\rangle dy \\
&-           \mu \int \frac{1}{(z-y)^2} %G(y;\z)
\langle   V_\gamma(y)  \prod_kV_{\alpha_k }(z_k)\rangle dy-\frac{_1}{^4} \mu^2 \gamma^2 \int \frac{1}{z-y} \frac{1}{z-x}\langle   V_\gamma(y) V_\gamma(x) \prod_kV_{\alpha_k }(z_k)\rangle
%G(y,x;\z)
dy dx%\label{tcontra2}
 \end{align*} 
Integrating  by parts we have
  \begin{align}
  -           \mu& \int \frac{1}{(z-y)^2} %G(y;\z)
  \langle  V_\gamma(y)   \prod_kV_{\alpha_k }(z_k)\rangle dy=
 -\mu\int\partial_{y} \frac{1}{z-y}  \langle V_\gamma(y)    \prod_kV_{\alpha_k }(z_k)\rangle %G(y;\z)
 dy=\mu\int \frac{1}{z-y}\partial_{y} 
 \langle   V_\gamma(y)\prod_kV_{\alpha_k }(z_k)\rangle %G(y;\z)
 dy\nonumber\\
 &=-\hf\gamma\mu\sum_i \alpha_i\int  \frac{1}{z-y}  \frac{1}{y-z_i} \langle   V_\gamma(y) \prod_kV_{\alpha_k }(z_k)\rangle %G(y;\z)
 dy+
 \frac{_1}{^2} \mu^2 \gamma^2 \int \frac{1}{z-y}\frac{1}{z-x} \langle   V_\gamma(y) V_\gamma(x) \prod_kV_{\alpha_k }(z_k)\rangle%G(y,x;\z)
 dy dx\nonumber\\
 &=-\hf\gamma\mu\sum_j \alpha_j \frac{1}{z-z_j} ( \int  \frac{1}{z-y} %G(y;\z)
 \langle   V_\gamma(y)  \prod_kV_{\alpha_k }(z_k)\rangle dy+ \int  \frac{1}{y-z_i} %G(y;\z)
 \langle   V_\gamma(y)  \prod_kV_{\alpha_k }(z_k)\rangle dy)\nonumber\\
 & +\frac{_1}{^4} \mu^2 \gamma^2 \int \frac{1}{z-y} \frac{1}{z-x} \langle   V_\gamma(y)V_\gamma(x)\prod_kV_{\alpha_k }(z_k)\rangle %G(y,x;\z)
 dy dx \label{ipp}
 \end{align} 
 so that 
 \begin{align*}
\langle T(z)\prod V_{\alpha_i}(\z)\rangle&=
(\frac{_Q}{^2}\sum_{i}\alpha_i  \frac{1}{(z-z_i)^2} -\frac{_1}{^4} \sum_{j,k} \frac{ \alpha_j \alpha_k}{(z-z_k)(z-z_j)}) \langle\prod_k  V_{\alpha_k }(z_k)\rangle%G(\z)
\nonumber \\& -\hf\mu  \gamma   \sum_{k} \alpha_k  \frac{1}{z-z_k} \int \frac{1}{y-z_k} \langle   V_\gamma(y) \prod_kV_{\alpha_k }(z_k)\rangle %G(y;\z)
dy \end{align*} 
On the other hand using \eqref{ipp1}
\begin{align}
 \partial_{z_i}\langle  \prod_k  V_{\alpha_k  }(z_k)\rangle& =
  \alpha_i\langle \partial_z \phi(z_i)\prod_k  V_{\alpha_k  }(z_k)\rangle\\&=-\hf\sum_{j\neq i}\frac{\alpha_i \alpha_j}{z_i-z_j} \langle\prod_k  V_{\alpha_k }(z_k)\rangle%G( \z) %\langle  \prod_k  V_{\alpha_k,\epsilon  }(z_k)\rangle_{\epsilon}
  +\hf\alpha_i\mu\gamma \int_\C\frac{1}{z_i-y}\langle  V_{\gamma }(y)\prod_k  V_{\alpha_k }(z_k)\rangle dy
 \label{Ydefi}
\end{align}
so that the 1st Ward identity follows.
%and we denoted for brevity $\frac{1}{(z-y)}_{\epsilon,\epsilon}$ by $\frac{1}{(z-y)}_{\epsilon}$.
%\vskip 2mm
Let us make some remarks regarding this calculation.
%\vskip 2mm

First, for the proof one needs to work with regularized correlations. Then some of the identities used in this calculation are not exact. Worse, some of the resulting integrals are only conditionally convergent. Using multiplicative chaos techniques one can study the divergence of the vertex opeartor correlations as two or more points come together. For instance for two points one gets
\begin{equation*}
\langle  V_{\gamma }(y)\prod_k  V_{\alpha_k  }(z_k)\rangle\leq C|y-z_i|^{-2+\delta}
 %\label{Ydefi}
\end{equation*}
with $\delta>0$. Hence this singularity is integrable (as is also evident from Lemma \ref{kpzid}). However above we need to control $(y-z_i)^{-1}$ times this and the result is {\it not} absolutely integrable. The clue what to do is in equation \eqref{ipp}.  The LHS is the Beltrami transform of the correlator computed at $z\neq z_i$. This is pointwise defined provided the correlator is H\"older continuous which can be shown by multiplicative chaos techniques. The first and third terms are absolutely convergent. As a result this identity  relates the potentially divergent integral to finite ones. For the proof one needs to work with a regularized version of the identity. As an upshot one obtains using \eqref {Ydefi} that the correlation functions are $C^1$.  For the second Ward identity one needs to control singular integrals such as
\begin{align*}
 \int \frac{1}{(z-y)^3} %G(y;\z)
  \langle  V_\gamma(y)   \prod_kV_{\alpha_k }(z_k)\rangle dy \end{align*} 
which are related by identities to less singular expressions. Upshot is that the correlations are $C^2$.

\section{Quantum Liouville Theory}

The Liouville model gives rise to {\it Quantum} Field Theory. This means in particular that there is a  canonical construction of a Hilbert space $\caH$ and a representation of the symmetries of the theory as operators acting on $\caH$. This reconstruction of quantum fields is very general and is based on a peculiar positivity property of the random field, the {\it reflection positivity} (or {\it Osterwalder-Schrader positivity}, \cite{OS}).  
\subsection{Liouville functional}
Recall that in the round metric $\hat g$  the Liouville field is given as
$$
\phi=c+X_{\hat g}+\frac{_Q}{^2}\log\hat g.
$$
and the Liouville "expectation" is given by
  \begin{align}\label{lexpe}
%\langle F\rangle=
\int F(\phi)d\nu_{\hat g}=\int dc\ e^{-2Qc}\E\ e^{-\mu \int_\C  e^{\gamma \phi}dz} F(\phi).
 \end{align}
 It will be convenient to make a change of variables from the zero average field $X_{\hat g}$ to one that has zero average on $\partial\D$. Let 
 $$m_{\partial\D}(X_{\hat{g}}):=\frac{{1}}{{2\pi}}\int X_{\hat g}(e^{i\theta})d\theta
 .$$ 
 Making a shift in the $c$-integral we get
 \begin{align*} 
 \int F(\phi)d\nu_{\hat g}
%\langle F \rangle
= \int dc e^{ -2Qc}\E e^{2Qm_{\partial\D}(X_{\hat{g}})} F(c+X_{\hat{g}} -m_{\partial\D}(X_{\hat{g}})+ Q/2\ln \hat{g})  
 e^{ -\mu
 \int_\C e^{\gamma (c+X_{\hat{g}}-m_{\partial\D}(X_{\hat{g}}) + Q/2\ln \hat{g}) }\,dz}
\end{align*}
By  the Girsanov theorem $X_{\hat g}$ under $e^{2Qm_{\partial\D}(X_{\hat{g}})-2Q^2 \E m_{\partial\D}(X_{\hat{g}})^2}\P_{\hat g}$ equals in law  $X_{\hat g}+2Q\frac{{1}}{{2\pi}}\int G_{\hat g}(z,e^{i\theta})d\theta$ under $\P_{\hat g}$. We have
$$\frac{{1}}{{2\pi}}\int G_{\hat g}(z,e^{i\theta})d\theta=k(z)-\frac{1}{4}\ln \hat{g}(z)+C_{\hat g}
$$
where we have  set
$$k(x)=\ln\frac{1}{|x|}\mathbf{1}_{\{|x|\geq 1\}}$$
 and 
 $ \E[m_{\partial\D}(X_{\hat{g}})^2]=C_{\hat g}$. Hence  
 we get the following expression after a further shift of $c$ by $-2QC_{\hat g}$
\begin{align}\label{FL} 
 \int F(\phi)d\nu_{\hat g}%\langle F \rangle
 = e^{ 6C_{\hat g} Q ^2  }\int dc e^{ -2Qc}\E F(c+\Phi+2Q k)
  e^{ -\mu
 \int_{\C}e^{\gamma (c+\Phi +2Q k) }\,dz}
\end{align}
We have defined 
the field
$$
\Phi=X_{\hat{g}}-m_{\partial\D}(X_{\hat{g}}).
$$
We have arrived to a new representation of the Liouville field as
\begin{align}\label{FL0} 
 \phi=c+\Phi +2Q k
\end{align}
where $\Phi$ is a gaussian field with covariance 
\begin{align*}%\label{eq:metricg}
G(z,z'):=\E \Phi(x)\Phi(y)=%&G_{ \hat{g}}(x,y)+\frac{1}{4}\ln \hat{g}(x)+\frac{1}{4}\ln \hat{g}(y)-\frac{1}{2}-\ln\frac{1}{|x|}\mathbf{1}_{\{|x|\geq 1\}}-\ln\frac{1}{|y|}\mathbf{1}_{\{|y|\geq 1\}}\\
\ln \frac{1}{|z-z'|}-k(z)-k(z'). %\ln\frac{1}{|x|}\mathbf{1}_{\{|x|\geq 1\}}-\ln\frac{1}{|y|}\mathbf{1}_{\{|y|\geq 1\}}.
\end{align*}
 We will construct the quantum theory starting with the  linear functional
\begin{align}\label{FL1} 
 \langle F \rangle
 = \int dc e^{ -2Qc}\E F(\phi)
  e^{ -\mu
 \int_{\C}e^{\gamma \phi }\,dz}
\end{align}
with $\phi$ given by \eqref{FL0}.
 
\subsection{Osterwalder-Schrader positivity}

 For $A\subset\C$ let ${\cal{F}}_A$  be the $\sigma$-algebra generated by $\int_\C\phi f$ with ${\rm supp} f\subset A$.  The Hilbert space is constructed out of  ${\cal{F}}_\D$. Let $\theta:\hat\C\to\hat\C$ be the reflection from the unit circle $\theta(z)=1/\bar z$.  Define $\Theta:{\cal{F}}_{\D}\to{\cal{F}}_{\D^c}$ by
   \begin{align}\label{lexpedefi}
\Theta F(\phi):=\overline{F(\theta \phi-2Q\ln|z|)}
 \end{align}
where $(\theta\phi)(z):=\phi(\theta z)=\phi(1/\bar z)$. Consider now the following sesquilinear form
  \begin{align}\label{scalar}
(F,G):=\langle \Theta F G\rangle.
 \end{align}
for $F,G\in\caF_\D$. OS-positivity is the following statement:

 %This follows from the following facts. First, the $\mu=0$ case i.e. the MFF is OS-positive:
 \begin{proposition}\label{OS} The form  \eqref{scalar} is  positive semidefinite:
  \begin{align}\label{scalar1}
\langle \Theta F F\rangle\geq 0.
 \end{align}
  %\int dc\ e^{-2Qc}\E\  F\Theta F\geq 0.
 \end{proposition} 
 The main ingredient in the proof is the corresponding statement for MFF i.e. $\mu=0$ case.Let
 $$
\langle F\rangle_0:=\int e^{-2Qc}\E\, F(\phi)dc.=\int e^{-2Qc}\E\, F(c+\Phi+2Qk)dc.
$$
We will decompose $\Phi$ to independent fields on $\D$, $\D^c$ and $S^1=\partial \D$. For this 
let $\Phi_\D(z)$ be the Dirichlet GFF on  $\D$, i.e. 
\begin{align}\label{diri}
G_\D(z,z'):=\E\Phi_\D(z)\Phi_\D(z')=\log\frac{|1-z\bar z'|}{|z-z'|}.
\end{align}
and $\Phi_{\D^c}(z)$ the  Dirichlet GFF on  $\D^c$ i.e.
$$
\Phi_{\D^c}\stackrel{law}{=} \theta \Phi_\D.
$$
%where 
%$$(\theta f)(z):=f(1/\bar z).$$
Next note that $\varphi:%\stackrel{law}
{=}\Phi|_{\partial\D}
$
 is the GFF on circle (with zero average)  i.e. concretely
$$
\varphi\stackrel{law}{=}\sum_{n\neq 0}\varphi_ne^{in\theta}\label{0}
$$
where 
$$\varphi_n=\frac{1}{2\sqrt n}(\alpha_n+i\beta_n)\ \ n>0,\ \ \ \varphi_{-n}=\bar\varphi_n
$$ with
$\alpha_n,\beta_n$ i.i.d. $N(0,1)$. 

Let 
$P\varphi$ be  the Harmonic extension of $\varphi$ defined on $\D$ by
\begin{align}\label{Paction}
Pe^{in\theta}=\bar z^n,\ n>0, \ \ \ Pe^{in\theta}=z^{-n},\ n<0,\ \ \ z=re^{-i\theta}\in \D.
\end{align}
i.e.
\begin{align}\label{Paction1}
(P\varphi)(z)=\sum_{n>0}(\varphi_n\bar z^n+\varphi_{-n}z^n)
\end{align}
On $\D^c$ $P\varphi$ is given by  
\begin{align}\label{Paction2}
P\varphi(z){=}(\theta P\varphi)(z),\ \ \ z\in \D^c.
\end{align}
The we have
\begin{proposition}
We may decompose as sum of independent fields:
$$
\Phi\stackrel{law}{=}\Phi_\D+P\varphi+\Phi_{\D^c}.
$$
\end{proposition}
\begin{proof} We have 
\begin{align*}
\E\varphi_n\varphi_m=\frac{1}{2|n|}\delta_{n,-m}
\end{align*}
so that for $z,u\in\D$
\begin{align*}
\E P\varphi(z)P\varphi(u)=\hf\sum_{n>0}((z\bar u)^n+(\bar z u)^n)=-\ln|1-z\bar u|
\end{align*}
and then for $z\in\D$, $u\in\D^c$
\begin{align*}
\E P\varphi(z)P\varphi(u)=-\ln|1-z/ u|=\ln |z-u|^{-1}-\ln|u|^{-1}
\end{align*}
It is then straightforward to check the equality of covariances.
\end{proof}
%For $A\subset\C$ let ${\cal{F}}_A$ consist of $F(X)$ measurable in $X|_A$. Define $\Theta:{\cal{F}}_{\D}\to{\cal{F}}_{\D^c}$ by
%$$
%\Theta F(X):=\overline{F(\theta X-2Q\theta k)}.
%$$
Using this decomposition we then get for  $F,G\in  {\cal{F}}_\D$:
\begin{align}
\langle  (\Theta F)G\rangle_0&=\int e^{-2Qc}\E (\Theta  F)(c+\Phi_{\D^c}+P\varphi+2Qk)G(c+\Phi_{\D}+P\varphi)dc\nonumber\\
&=\int e^{-2Qc}\E_\varphi(\overline{ \E_\D F(c+\Phi_\D+P\varphi)}\E_\D G(c+\Phi_{\D}+P\varphi))dc.\label{1}
\end{align}
Hence
\begin{align}
\langle  (\Theta F) F\rangle_0\geq 0.\label{1cor}
\end{align}
The Proposition \ref{OS} follows then from

\begin{lemma}Let $I(X):= \int_\D  e^{\gamma \phi(z)}dz$. Then
  \begin{align}\label{factor}
  \int_\C  e^{\gamma \phi}dz=I(\phi)+(\Theta I)(\phi)
 \end{align}
\end{lemma}
\begin{proof} In the same way as Proposition \ref{mobiusforchaos} one proves the change of variables
formula
$$
 (e^{\gamma \theta \phi})(z)=%e^{\gamma (c+\frac{_Q}{^2}\log\hat g(1/\bar z))}
 |z|^{\gamma^2}(e^{\gamma \phi})(1/\bar z)
$$
and thus
$$
(\Theta I)(\phi)=\int_{\D} |z|^{\gamma^2}(e^{\gamma \phi }(1/\bar z))e^{-2\gamma Q\ln|z| }dz=
\int_{\D} |z|^{-4}(e^{\gamma \phi}  (1/\bar z))dz=\int_{\D^c} e^{\gamma \phi(z)}dz
$$
\end{proof}

\subsection{Hilbert space}

The  Liouville Hilbert space $\caH$ is defined as the completion of  $\caF_\D/\caN$ where
$$\caN=\{F\in\caF_\D|(F,F)=0\}.$$
 Then
  \begin{align}\label{scalar}
(F,G)=( UF, UG)_{ L^2(d\P(\varphi)dc)}
 \end{align}
 where
 \begin{align}\label{Utilde}
 UF:=e^{-Qc}\E_\D e^{-\mu \int_\D  e^{\gamma X}dz}F%=Ve^{-\int_0^\infty v(X_s)ds}F.
 \end{align}
 Thus 
 $$ U:\caH\to L^2(d\P(\varphi)dc)
 $$ is an isometry and we may identify $\caH$ with  a subspace of $ L^2(d\P(\varphi)dc)$.

\subsection{Q-Free Field}

Let us consider the $\mu=0$ case in more detail. This is the Free field with "background charge" $iQ$.
We may realize $ L^2(d\P(\varphi)\, dc)$ as
$$
d\P(\varphi)=\prod_{n>0}\frac{_1}{^{2\pi}}e^{-\hf(\alpha_n^2+\beta_n^2)}d\alpha_nd\beta_n.%=\prod_{n>0}e^{-2n \varphi_n\varphi_{-n }}(-2ni)d\varphi_n\wedge d\varphi_{-n}.
$$
We have then 
\begin{proposition} \label{Ufree} The map %$U:{\cal F}_\D  L^2(d\P(\varphi)dc )$
$$
(UF)(c,\varphi)=e^{-Qc}\E_\D F(c+\Phi_D+P\varphi).
$$
extends to  a unitary map  $U:\caH\to L^2(d\P(\varphi)\, dc)$.
\end{proposition}
\begin{proof} 
By \eqref{1} $U$  is an isometry  from $\caF_\D$ to a subspace of $ L^2(d\P(\varphi)\, dc)$. To show $U$ is onto note that $c=\frac{1}{2\pi}\int \phi(e^{i\theta})d\theta$ and %Let $\psi\in L^2(\R)$ and c
consider $F$ of the form
\begin{align}
F(\phi)=\psi(\int \phi(e^{i\theta})\frac{_{d\theta}}{^{2\pi}})
{e^{( \phi,f)}}=\psi(c)e^{c( 1,f)}{e^{( \Phi,f)}}\label{separable}
\end{align}
where $f\in C_0^\infty(\C)$ and we use in this Chapter the notation 
$( \phi,f):=\int_\D \phi f$. 
Then
\begin{align}
(UF)(c,\varphi))=\psi(c)e^{-Qc}e^{c( 1,f)+\hf( f,G_\D f)}e^{( P\varphi,f)}
\label{separable1}
\end{align}
From \eqref{Paction1} we get
$$
( P\varphi,f)=\sum_{n>0}(\varphi_n\int_\D\bar z^n f+\varphi_{-n}\int_\D z^n f).%:=\sum_{n\neq 0}\varphi_n\pi(f)_{-n}:=(\varphi,\pi(f))
$$
so that by \eqref{0}% and \eqref{Paction} we get 
\begin{align}
( P\varphi,f)=\sum_{n\neq 0}\varphi_n\pi(f)_{-n}:=(\varphi,\pi(f))
\end{align}
where we defined for $n>0$:
\begin{align}
\pi(f)_n=\int_\D z^n f, \ \ \ \pi(f)_{-n}=\int_\D \bar z^n f
\end{align}
and we denote the scalar product in $L^2(\partial\D)$ also by $(\cdot,\cdot)$.
%Taking $f=\frac{1}{\pi}(Q+ip)+ig$ with $\int_\D g=0$ we get
%Let $f_0:=( 1,f)$. 
Thus
$$
UF=\psi(c)e^{(( 1,f)-Q)c}e^{%ipc
\hf( f,G_\D f)}e^{( \varphi,\pi(f))}.
$$
The linear span of such functions is dense in  $ L^2(d\P(\varphi)\, dc)$.
\end{proof}
In particular, for the vertex operator
$$
V_\alpha(z)=e^{\alpha \phi(z)-\hf\alpha^2 \E\Phi(z)^2}
$$
we get
%Defining for a Gaussian variable $u$ $:e^u:=e^u/\E e^u$ we have
\begin{align}
UV_\alpha(z)%\frac{e^{( X,f)}}{\E e^{( \Phi,f)}}
=e^{(\alpha-Q)c}:e^{\alpha (P\varphi)(z)}:.\label{vertex}
\end{align}
where 
$$
:e^{\alpha (P\varphi)(z)}:=e^{\alpha (P\varphi)(z)-\hf\alpha^2\E (P\varphi)(z)^2}=(1-|z|^2)^{\hf\alpha^2}e^{\alpha (P\varphi)(z)}
$$

\subsection{Hamiltonian of Q-free field }
From now on we use for $\phi$  the representation
$$
\phi=%c+X_{\hat g}+\frac{_Q}{^2}\log\hat g\stackrel{law}{=}
c+\Phi+2Qk.
$$
For $q\in\C$ define dilation
$$
(s_qf)(z)=f(qz)
$$
and
$$
S_qF(\phi)=F(s_q\phi+Q\log|q|).
$$
Hence $S_q:\caF_\D\to\caF_\D$ if $|q|\leq 1$. We have
\begin{proposition}  The adjoint  of $S_q$ is $S_q^\ast=S_{\bar q}$ i.e.
$$
(F,S_qG)=(S_{\bar q}F,G).
$$
\end{proposition}

\begin{proof} 
By M\"obius invariance of the Liouville expectation, Corollary \ref{mobius}, we get
\begin{align*}
(F,S_qG)&=\langle (\Theta F)(\phi)(S_qG)(\phi)\rangle=\langle (\Theta F)(\phi)G(s_q\phi+Q\log|q|)\rangle=\langle (\Theta F)(s_{q^{-1}}\phi-Q\log|q|)G(\phi)\rangle% \nonumber\\ &=\langle F(\theta(s_{q^{-1}}\phi-2Qk)-Q\log|q|)G(\phi)\rangle%\nonumber\\
%&=\langle F(s_{\bar q}(\theta(\phi-2Qs_qk)-Q\log|q|)G(\phi)\rangle\nonumber\\&=\langle (\Theta S_{\bar q} F)(\phi-2Qs_qk+2Qk-2Q\log|q|)G(\phi)\rangle\nonumber\\&=(S_{\bar q}F,G)
\end{align*}
By  the definition of $\Theta$, \eqref{lexpedefi} we may write this as
\begin{align}
(F,S_qG)%&=\langle (\Theta F)(\phi)(S_qG)(\phi)\rangle=\langle (\Theta F)(\phi)G(s_q\phi+Q\log|q|)\rangle=\langle (\Theta F)(s_{q^{-1}}\phi-Q\log|q|)G(\phi)\rangle \nonumber\\ 
&=\langle \bar F(\theta(s_{q^{-1}}\phi-2Qk)-Q\log|q|)G(\phi)\rangle\nonumber\\
&=\langle \bar F(s_{\bar q}(\theta(\phi-2Qs_qk)-Q\log|q|)G(\phi)\rangle\nonumber\\&=\langle (\Theta S_{\bar q} F)(\phi-2Qs_qk+2Qk-2Q\log|q|)G(\phi)\rangle\nonumber\\&=(S_{\bar q}F,G)\label{sym}
\end{align}
since $\Theta S_{\bar q} F$ is supported in $q^{-1}\D^c$ and $1_{q^{-1}\D^c}(k-s_qk)=\log|q|$.
\end{proof} 
$S_q$ gives rise two semi groups. Taking $q=e^{-t}$ we define
$$
US_{e^{-t}}F=e^{-tH}UF.
$$
This is a contraction semigroup with generator $H\geq 0$, the {\it Hamiltonian} of the GFF. Taking $q=e^{i\alpha}$ we define
$$
US_{e^{i\alpha}}F=e^{i\alpha P}UF.
$$
where $P$ is the {\it momentum} operator of the GFF. It is a generator of an unitary group. To compute them explicitly we use the complex coordinates $\{\varphi_n\}_{n\in\Z}$  in $L^2(\P(d\varphi))$ given in the representation \eqref{0}
%$$\mathfrak{h}:=L^2(\P(d\varphi))=L^2(\prod_{n>0}\frac{_1}{^{2\pi}}e^{-\hf(\alpha_n^2+\beta_n^2)}d\alpha_nd\beta_n).
%$$
 and 
define  for $n>0$:
\begin{align*}
a_n=\hf\frac{\partial}{\partial\varphi_{-n}},\ \ \ 
a_{-n}=n\varphi_{-n}-\hf \frac{\partial}{\partial\varphi_{n}}
\end{align*}
and
\begin{align*}
\tilde a_n=\hf\frac{\partial}{\partial\varphi_{n}},\ \ \ 
\tilde a_{-n}=n\varphi_{n}-\hf \frac{\partial}{\partial\varphi_{-n}}.
\end{align*}
$a_n$ and $\tilde a_n$ are ($n>0$) called the {\it  annihilation operators}  (for analytic and anti analytic modes) and $a_{-n}$ and $\tilde a_{-n}$ the {\it creation operators}. They are densely defined closable operators in  $L^2(\P(d\varphi))$ and their closures satisfy $a_n^\ast=a_{-n}$, $\tilde a_n^\ast=\tilde a_{-n}$. Furthermore we have
 $a_n1=0$ and $\tilde a_n1=0$ for  $n>0$ and we have the commutation relations
\begin{align*}
[a_n,
a_{m}]=\frac{_n}{^2}\delta_{n,-m}=[\tilde a_n,
\tilde a_{m}],\ \ 
[a_n,\tilde a_m]=0.
\end{align*}
%RATHER DEFINE $a_n=\frac{\partial}{\sqrt{2}\partial\varphi_{-n}}$ and

\begin{proposition} \label{hamilton} We have
\begin{align*}
H&=\hf(-\frac{d^2}{dc^2}+Q^2)+2\sum_{n>0}(a_{-n}a_n+\tilde a_{-n}\tilde a_n)\\
P&=2\sum_{n>0}(a_{-n}a_n-\tilde a_{-n}\tilde a_n)
\end{align*}
%$$
%US_qF=
%q^{\caL_0}\bar q^{\tilde \caL_0}UF.$$
%where
%\begin{align*}
%\caL_0=\frac{_1}{^4}(- \partial_c^2+Q^2)+\sum_{n>0}a_{-n}a_n,\ \ \ \tilde\caL_0=\frac{_1}{^4}(- \partial_c^2+Q^2)+\sum_{n>0}\tilde a_{-n}\tilde a_n
%\end{align*}

\end{proposition}
\begin{proof} 

It suffices to compute
$$
US_qF=e^{-Qc}\E_\D F(c+s_q\Phi_\D+s_qP\varphi+Q\log|q|)
$$
for 
$$F=e^{(\phi,f).%-\hf\E(\Phi,f)^2
}%=e^{(X,f)}.%-\hf ( f,Gf ) }
$$
 %given by \eqref{separable}.
We get 
$$US_qF=e^{-Qc}e^{(c+Q\log|q|)( 1,f)}%\frac
\E_\D{e^{( \Phi_\D+P\varphi ,f_q)}}=e^{-Qc}e^{(c+Q\log|q|)( 1,f)+\hf (f_q,G_\D f_q)}e^{(\varphi,\pi(f_q))}.%e^{-\hf ( f,Gf ) }%{\E e^{( \Phi_D+P\varphi,f)}}
$$ with 
$$f_q(z)=|q|^{-2}f(z/q).$$ 
Now
%\begin{align}
%\E_\D{e^{( \Phi_\D+P\varphi ,f_q)}}%e^{-\hf ( f,Gf ) }
%\frac{\E_\D e^{( \Phi_\D+P\varphi ,f_q)}}{\E e^{( \Phi_D+P\varphi,f)}}
%&=e^{\hf (f_q,G f_q)-\hf (f,G f)}:e^{(\varphi,\pi(f_q))}:\\
%&=e^{-\hf\log|q| (1,f)^2}:e^{(\varphi,\pi(f_q))}:
%=e^{\hf (f_q,G_\D f_q)}e^{(\varphi,\pi(f_q))}.
%\end{align}
 from \eqref{diri}
$$
(f_q,G_\D f_q)=(f,G_\D f)-\ln |q|(1,f)^2+\int \log\frac{|1-\bar q qz\bar z'|}{|1-z\bar z'|}f(z)f(z')dzdz'.
$$
We compute next $\partial_qUS_qF$ and $\partial_{\bar q}US_qF$ at $q=1$. First
\begin{align*}
\partial_q|_{q=1}(f_q,G_\D f_q)&=-\hf(1,f)^2-\hf\int (\frac{ z \bar z'}{1-z\bar z'}+\frac{z' \bar z}{1-\bar z z'})f(z)f(z')dzdz'\\&
=-\hf(1,f)^2-\sum_{n>0}\pi(f)_n\pi(f)_{-n}.
\end{align*}
Next,
$$
\pi(f_q)_n=\int_\D z^nf_q=q^n\pi(f)_n, \ \ n>0
$$
and $$
\pi(f_q)_n=\int_\D \bar{z}^{-n}f_q=\bar q^{-n}\pi(f)_n, \ \ n<0
$$
so that 
$$
\partial_q|_{q=1}\pi(f_q)_n=n\pi(f)_n1_{n>0}.
$$
Altogether we get
\begin{align*}
\partial_q|_{q=1}US_qF&=(\hf Q(1,f)-\frac{_1}{^4} (1,f)^2+\sum_{n>0}((n\pi(f)_n\varphi_{-n}-\hf\pi(f)_n\pi(f)_{-n})UF
\end{align*}
Using $ (\partial_c+Q)UF=(1,f)UF$ and $\frac{\partial}{\partial\varphi_{n}}UF=\pi(f)_{-n}UF$ this becomes
\begin{align*}
\partial_q|_{q=1}UF&=%(\hf Q(\partial_c+Q)-\frac{_1}{^4}  (\partial_c+Q)^2
\frac{_1}{^4}(- \partial_c^2+Q^2)+\sum_{n>0}(n\varphi_{-n}-\hf\frac{\partial}{\partial\varphi_{n}})\frac{\partial}{\partial\varphi_{-n}}
UF\\&
=\frac{_1}{^4}(- \partial_c^2+Q^2)+2\sum_{n>0}a_{-n}a_n
\end{align*}
In the same way, using
$$
\partial_{\bar q}|_{q=1}\pi(f_q)_n=-n\pi(f)_n1_{n<0}.
$$
we get
\begin{align*}
\partial_{\bar q}|_{q=1}UF&=%(\hf Q(\partial_c+Q)-\frac{_1}{^4}  (\partial_c+Q)^2
\frac{_1}{^4}(- \partial_c^2+Q^2)+\sum_{n>0}(n\varphi_{n}-\hf\frac{\partial}{\partial\varphi_{-n}})\frac{\partial}{\partial\varphi_{n}}
UF\\&
=\frac{_1}{^4}(- \partial_c^2+Q^2)+2\sum_{n>0}\tilde a_{-n}\tilde a_n
\end{align*}
%\begin{align*}
%\partial_q|_{q=1}US_qF&=&e^{-\hf\log|q| (1,f)^2}e^{-Qc}e^{(c+Q\log|q|)( 1,f)}
%  :e^{(\varphi,\pi(f_q))}: .
%\end{align*}
\end{proof}

\subsection{Virasoro algebra}
The  Energy-Momentum  tensor field is given by
$$
T(z)=Q\partial_z^2 \phi\,  -((\partial_z \phi)^2-\E( \partial_z\phi)^2).
$$
Let  the support of $F$ be in the ball $B_r$ of radius $r<1$ centred at origin. We define for $n\in\Z$
\begin{align*}
L_nUF:&=\frac{1}{2\pi i}U\oint z^{n+1}T(z)F\\
\bar L_nUF:&=\frac{1}{2\pi i}U\oint \bar z^{n+1}T(z)F
\end{align*}
where the contour is in $B_r^c$. Let us compute these operators explicitely on a dense domain. Again, it suffices to take $F=e^{(\phi,f)}$ with $f$ supported in $B_r$. We compute first $U\partial_z\phi(z)F$. First note that
\begin{align*}
U((\phi,g)F)&=\partial_\lambda|_0Ue^{(\phi,f+\lambda g)}=(c(1,g)+(P\varphi,g)+(g,G_\D f))UF.
\end{align*}
Now take $g=-\partial\delta_z$ so that
 \begin{align*}
U\partial \phi(z)F&=(\partial P\varphi(z)+\partial (G_\D f)(z))UF.
\end{align*}
From \eqref{diri} %and \eqref{0}
 we get
\begin{align*}
\partial_zG_\D (z,u)=-\hf(\frac{1}{z-u}-\frac{1}{z-\frac{1}{\bar u}})=-\hf\sum_{n=0}^\infty(u^nz^{-n-1}+\bar u^{n+1}z^n)
\end{align*}
which converges since $|u|<|z|<1$. Hence we get
\begin{align*}
\partial (G_\D f)(z)=-\hf\sum_{n=0}^\infty(z^{-n-1}\pi(f)_n+z^n\pi(f)_{-n-1})=-\hf\sum_{n\in\Z}\pi(f)_{-n-1}z^n
\end{align*}
where we defined $\pi(f)_0=(1,f)$. Since 
\begin{align*}
\partial P\varphi(z)=\sum_{n=1}^\infty nz^{n-1}\varphi_{-n}
\end{align*}
we end up with
\begin{align*}
U\partial \phi(z)F&=\sum_{n\neq 0}z^{n-1}(n\varphi_{-n}1_{n>0}-\hf\partial_{\varphi_n})UF-\hf(\partial_c+Q)UF=\sum_{n\in\Z}z^{-n-1}a_nUF
\end{align*}
provided we define
\begin{align*}
a_0=-\hf(\partial_c+Q).
\end{align*}
Next consider $U:(\partial \phi(z))^2:F$ where $:(\partial \phi(z))^2:=(\partial \phi(z))^2-\E_\D(\partial \Phi(z))^2- \E(\partial P\varphi(z))^2$. Noting that $\E(\partial P\varphi(z))^2=0$ we then compute
\begin{align*}
U:(\partial \phi(z))^2:F&=\sum z^{-n-1}z^{-m-1}(a_na_m+\hf m\delta_{n,-m})UF=\sum z^{-n-1}z^{-m-1}:a_na_m:UF
%((\sum_{n\in\Z}z^{n-1}(n\varphi_{-n}1_{n>0}-\hf\pi(f)_{-n}))^2-\E(\partial P\varphi(z))^2)UF.
\end{align*}
where $:a_na_m:=a_na_m$ if $m>0$ and $a_ma_n$ if $n>0$ (i.e. annihilation operators are on the left). Combining we then get
\begin{align*}
L_n=-(n+1)Qa_n-\sum_{m\in\Z}:a_{n-m}a_m:.
%((\sum_{n\in\Z}z^{n-1}(n\varphi_{-n}1_{n>0}-\hf\pi(f)_{-n}))^2-\E(\partial P\varphi(z))^2)UF.
\end{align*}
In particular we get 
\begin{align*}
L_0=-Qa_0-a_0^2-2\sum_{n>0}a_{-n}a_n=\frac{_1}{^4}(-\partial_c^2+Q^2)-2\sum_{n>0}a_{-n}a_n.
%((\sum_{n\in\Z}z^{n-1}(n\varphi_{-n}1_{n>0}-\hf\pi(f)_{-n}))^2-\E(\partial P\varphi(z))^2)UF.
\end{align*}
Comparing with Proposition \ref{hamilton}  we get
\begin{align*}
H&=L_0+\tilde L_0\\
P&=L_0-\tilde L_0
\end{align*}
$L_n$ satisfy the commutation relations of the {\it Virasoro Algebra}:
\begin{align*}
[L_n,L_m]=(n-m)L_{n+m}+\frac{c_L}{12}(n^3-n)\delta_{n,-m}
\end{align*}
where the Central charge is
\begin{align*}
c_L=1+6Q^2.
\end{align*}
The operators $L_n$ are densely defined  in $\caH$  and closable. They satisfy $L_{-n}=L_n^\ast$ i.e. we have a unitary representation $\caR$ of the Virasoro Algebra on $\caH$.

\subsection{Spectrum}

By Fourier transform in the $c$ variable we represent $\Psi\in\caH$ as $\{\hat\Psi(p)\}_{p\in \R}$ with
$\hat\Psi(p)\in L^2(d\P)$:
\begin{align*}
(\Psi,\Psi)=\int_\R dp\E|\hat\Psi(p)|^2
\end{align*}
We have 
\begin{align*}
\widehat{(a_0\Psi)}%^{\hat a}
(p)=-\hf(Q+ip)\hat\Psi(p)
\end{align*}
Hence our representation is reducible
$$
\caR=\int_\R^\oplus\caR_p.
$$ 
We can formally write this as
\begin{align*}
\Psi=\int_\R dp \hat\Psi(p)|p\rangle
\end{align*}
with 
\begin{align*}
\langle p,p'\rangle=\delta(p-p')
\end{align*}
and $|p\rangle$ is a highest weight state 
\begin{align*}
L_0|p\rangle=\Delta_p|p\rangle,\ \ \ L_n|p\rangle=0, n>0
\end{align*}
where $\Delta_p=\frac{1}{4}(Q^2+p^2)$. From \eqref{vertex} we get
\begin{align*}
|p\rangle=UV_{Q+ip}(0).
\end{align*}

\subsection{Liouville Hamiltonian}
Recall that for $\mu>0$
 \begin{align}\label{Utilde}
 UF:=e^{-Qc}\E_\D e^{-\mu \int_\D  e^{\gamma \phi}dz}F%=Ve^{-\int_0^\infty v(\phi_s)ds}F.
 \end{align}
 provides an isometry
 $$ U_L:\caH\to L^2(d\P(\varphi)dc).
 $$% is an isometry and we may identify $\caH_L$ with  $ L^2(d\P(\varphi)dc)$. 
 which allows us to identify $\caH$ with a subspace of $L^2(d\P(\varphi)dc)$. The Liouville semigroup is defined by
 $$
e^{-tH_L}U=US_{e^{-t}}.
$$
Since the  Liouville functional also satisfies \eqref{mobius} the calculation \eqref{sym} may be repeated to conclude
\begin{align}
(S_qF,G)=(F,S_{\bar q}G)\label{symL}
\end{align}
\begin{proposition} $e^{-tH_L}$ is a contraction semigroup on $ \caH$ with a positive
generator $H_L$.

\end{proposition}
 \begin{proof} %Let $\caS_\D=\{f\in C_0^\infty(\D): (1,f):=\int_\D fd\lambda_{\hat g}=0\}$ and consider observables of the form
 %$$
 %F=\sum_{i=1}^na_i((X,1))e^{(X,f_i)}=\sum_{i=1}^na_i(c+\frac{_Q}{^2}(\log\hat g,1))e^{(X_{\hat g}+\frac{_Q}{^2}\log\hat g,f_i)}
 %$$
 %where $a_i\in L^2(\R)$ and $f_i\in\caS_\D$. 
({\it scetch}). Let $F\in \caF_\D$ satisfy $ \langle F^2\rangle<\infty$. Such $F$ form a  dense set in $\caH$. Let $\psi=U_LF$. We have
 \begin{align}
 \|S_{e^{-t}}F\|&=(S_{e^{-t}}F,S_{e^{-t}}F)^\hf=(F,S_{e^{-2t}}F)^\hf\leq \| F\|^\hf\|S_{e^{-2t}}F\|^\hf \leq\dots\leq  \|F\|^{1-2^{-k}}\|S_{e^{-2^kt}}F\|^{2^{-k}}\nonumber
 \end{align}
Now
 $$
 \|S_{e^{-2^kt}}F\|^{2}= \langle S_{e^{-2^kt}}F\Theta S_{e^{-2^kt}}F  \rangle\leq \langle(S_{e^{-2^kt}}F)^2\rangle
 = \langle F^2\rangle
 $$
 by conformal invariance of $\langle \cdot\rangle$. Taking $k\to\infty$ we conclude $ \|S_{e^{-t}}F\|
 \leq \|F\|$.
 \end{proof}
  
\subsection{Feynman-Kac formula}
In this section we proceed heuristically.  Let 
 $$\phi_\tau(\theta)=\phi(e^{-\tau+i\theta})-Q\tau
 $$ 
 and consider observable
 $$
 F_\tau(\phi)=f(\phi_\tau).
 $$
 Then
 $$
 S_{e^{-t}}F_\tau=F_{\tau+t}.
 $$
 Taking $\tau=0$ we get %(recall \eqref{Vdef})
 $$
 UF_t=e^{-tH_L}f.
 $$
 Recall that
 $$
 e^{\gamma \phi}:=\lim_{\epsilon\to 0}\epsilon^{\frac{\gamma^2}{2}}e^{\gamma \phi_\epsilon}.
 $$
 We get, taking into account scaling of $\epsilon$
 $$
\mu \int_\D  e^{\gamma \phi}dz=\mu\int_0^\infty dt\int_0^{2\pi} d\theta e^{-(2+\frac{\gamma^2}{2})\tau}e^{\gamma \phi(e^{-t+i\theta})}=\mu\int_0^\infty dt\int_0^{2\pi}  d\theta e^{\gamma \phi_t(\theta)}:=\int_0^\infty  v(\phi_t)dt.
 $$
 Now we get by Trotter product formula, denoting by $U_0$ the $\mu=0$ map in Proposition \ref{Ufree} 
 \begin{align*}
 U_0e^{-\int_0^t v(\phi_s)ds}S_{e^{-t}}&=\lim_{n\to\infty}U_0e^{-\sum_i\frac{t}{n}v(\phi_{\frac{_{it}}{^n})}}S_{e^{-t}}=\lim_{n\to\infty}U_0(e^{-\frac{t}{n}v(\phi_{0})}S_{e^{-t/n}})^n\\
 &=\lim_{n\to\infty}(e^{-\frac{t}{n}v(\phi_{0})}e^{-t/nH})^nU_0=e^{-t(H+v)}U_0.
 \end{align*}
 Hence in particular
 \begin{align}\label{semig1}
 U_0e^{-\int_0^T v(\phi_s)ds}S_{e^{-t}}F=U_0e^{-\int_0^t v(\phi_s)ds}S_{e^{-t}}Fe^{-\int_0^{T-t} v(\phi_s)ds}=e^{-t(H+v)}(U_0e^{-\int_0^{T-t} v(\phi_s)ds}F).
 \end{align}
  Taking $T\to\infty$ in \eqref{semig1} we then get
  \begin{align}\label{semig2}
U S_{e^{-t}}=e^{-t(H+v)}U
 \end{align}
 i.e. the Liouville Hamiltonian is given by
 \begin{align}
 H_L=H+v%=\hf(-\frac{d^2}{dc^2}+Q^2)+\sum_{n>0}n(a^\ast_na_n+b^\ast_nb_n)+v
 \end{align}
 This is only formal. What is the precise definition of $v$? Formally it is just  
 $$
 v=\int  :e^{\gamma\varphi(\theta)}:d\theta.
 $$
Note that 
$$
\E :e^{\gamma\varphi(\theta)}: :e^{\gamma\varphi(\theta')}:=|e^{i\theta}-e^{i\theta'}|^{-\gamma^2}.
$$
Hence $v$ is not defined in the Fock space for $\gamma\geq 1$ since $\|v1\|=\infty$.  
% \end{document}

\subsection{Representation Theory} 
  Let $\caV$ be the linear span of the vectors 
  $
  U( \prod_{i=1}^n V_{\alpha_i}(z_i))
  $ with $|z_i|<1$. Then 
$$
L_n=\oint_{|z|=r} z^{n+1}T(z).
$$
acts on $\caV$ by taking $1-r$ small enough. % and using the Ward identities. %By analyticity the result does not depend on $r$.
The
Ward identities imply the  {\it Virasoro algebra} commutation rules on $\caV$:
$$
[L_m,L_n]=(m-n)L_{m+n}+\frac{_{c_L}}{^{12}}\delta_{m,-n}.
$$
The operators satisfy $L_n^\ast=L_{-n}$ on $\caV$. A major challenge is to find a common dense domain in $\caH$ for these operators so that the representation is unitary and then study its reduction to irreducibles. It is conjectured \cite{rib} that
$$\caH=\int_{\R_+}^\oplus \caH_P\ dP$$
where $\caH_P$ is a highest weight module $M_P={\rm span}\{L_{n}\psi_P, n\leq 0\}$,  $L_0\psi_P=\Delta_{Q+iP}\psi_P$, $L_n\psi_P=0$, $n>0$. % and h   $\alpha=Q+iP$ occurs with multiplicity one
As in the $\mu=0$ case  $\psi_P$ would be a generalized eigenfunction for $L_0$ i.e. not a vector in $\caH$. It  formally corresponds to the vertex operator $V_{Q+iP}$ which saturates the Seiberg bound. In \cite{DKRV2} these were constructed for $P=0$. It would be nice to understand the complex case.

Note that the spectrum for Liouville is $\R_+$ and not $\R$ as for the MFF. This is due to the potential barrier at positive $c$. Consider a toy model where we keep only the $c$ degree of freedom i.e. the operator
\begin{align*}
H&=\hf(-\frac{d^2}{dc^2}+Q^2)+\mu e^{\gamma c}
\end{align*}
on $L^2(\R)$. The spectrum of $H$ is determined by the $c\to -\infty$ asymptotics of $H$ i.e. it is $ [\hf Q^2,\infty)$ like in the $\mu=0$ case.  In the $\mu=0$ case the spectrum has  degenarcy two: the (generalized) eigenfunctions are $e^{\pm ipc}$.  In the $\mu>0$ case 
as $c\to -\infty$ the eigenfunctions are linear combinations of these plane waves but as  $c\to \infty$ they have to vanish. This means the degeneracy is one and asymptotics at $c\to -\infty$ is $\psi_p(c)=e^{ipc}+R(p)e^{-ipc}$ and $p\in\R_+$. It is a challenge to extend this analysis to the full Liouville Hamiltonian.

\subsection{DOZZ-conjecture} In conformal field theory it is believed \cite{BPZ} that all correlation functions are determined by the knowledge of primary fields (i.e. spectrum of representations) and their three point functions. For the latter  there is a remarkable conjecture due to Dorn, Otto, Zamolodchikov and  Zamolodchikov \cite{Do,ZZ} in Liouville theory. By M\"obius invariance 
$$
\langle V_{\alpha_1}(z_1)V_{\alpha_2} (z_2)V_{\alpha_3}(z_3)\rangle
= |z_1-z_2|^{2 \Delta_{12}} |z_2-z_3|^{2 \Delta_{23}} |z_1-z_3|^{2 \Delta_{13}} C_{\gamma}( \alpha_1, \alpha_2, \alpha_3)
$$
where  $\Delta_{12}= \Delta_{\alpha_3}-\Delta_{\alpha_1}-\Delta_{\alpha_2}  $ etc. The {\it three point structure constants} and be obtained from 
\begin{align}
C_{\gamma}(\alpha_1,\alpha_2,\alpha_3)&=\lim_{z_3\to\infty} |z_3|^{4 \Delta_3} \langle       V_{\alpha_1}(0)  V_{\alpha_2}(1) V_{\alpha_3}(z_3) \rangle \label{Climit}
\end{align}
which leads to the following expression in terms of Multiplicative chaos
$$
C_{\gamma}( \alpha_1, \alpha_2, \alpha_3)=%e^{\frac{1}{4}(s^2+2Qs)+2\ln 2 \Delta(\alpha_1)} \gamma^{-1} 
const.\mu^{-s}\Gamma(s)\,
\E\,  Z^{-s}
$$
with 
$$
Z=\int |z|^{-\alpha_1\gamma}|z-1|^{-\alpha_2\gamma}\hat g(z)^{-\frac{\gamma}{4}\sum_{i=1}^{3}\alpha_i}M_{\hat g,\gamma}(dz)
.$$ 
The   {DOZZ Conjecture} gives an explicit formula for $C_{\gamma}( \alpha_1, \alpha_2, \alpha_3)$. It is based on analyticity and symmetry  assumptions  that lack proofs. One of the ingredients in its derivation was recently proved in \cite{krv} namely  the so-called {\it BPZ equations} \cite{BPZ})  for the vertex operators $V_{- \frac{\gamma}{2}}$ and $V_{- \frac{2}{\gamma}}$ (in the language of CFT, $V_{- \frac{\gamma}{2}}$ and $V_{- \frac{2}{\gamma}}$ are called a degenerate fields). More precisely, setting 
$F(z,{\bf z}):=\langle   V_{- \frac{\gamma}{2}}(z) \prod_i V_{\alpha_i}(z_i)   \rangle,$ we prove
\begin{align*}
 (\frac{_4}{^{\gamma^2}}\partial_{z}^2   + \sum_k (\frac{\Delta_{\alpha_k}}{(z-z_k)^2}   +  %\sum_k 
 \frac{1}{z-z_k}  \partial_{z_k}) )F     =  0.%, \quad  \textbf{BPZ equations} . 
\end{align*}
and the same equation for $\langle   V_{- \frac{2}{\gamma}}(z) \prod_i V_{\alpha_i}(z_i)   \rangle$ with $\frac{4}{{\gamma^2}}$ replaced by $\frac{{\gamma^2}}{4}$.

The proof of these equations proceeds as the proof of the Ward identities: differentiation of the correlation function brings down a $\partial_z\phi(z)$ which then can be integrated by parts in the Gaussian measure. Note that we are not postulating that e.g. $V_{- \frac{\gamma}{2}}(z)$ is a degenerate field but rather proving it in the sense that it satisfies the expected equation. 

Using the BPZ equation, we recover an explicit  formula found earlier in the physics literature for the 4 point correlation functions $\langle    V_{-\frac{\gamma}{2}}  (z)  \prod_{i=1}^3 V_{\alpha_i}(z_i)  \rangle$ and  $\langle    V_{-\frac{2}{\gamma}}  (z)  \prod_{i=1}^3 V_{\alpha_i}(z_i)  \rangle$ in terms of the 3-point structure constants.   Any four point function is fixed by M\"obius invariance up to a single function depending on the cross ratio of the points. Specializing to the case we are interested in, we get
\begin{align}
  \langle    V_{-\frac{\gamma}{2}}  (z)  \prod_{l=1}^3 V_{\alpha_l}(z_l)  \rangle 
 & = |z_3-z|^{- 4 \Delta_{-\frac{\gamma}{2}}}  |z_2-z_1|^{ 2 (\Delta_3-\Delta_2-\Delta_1-\Delta_{-\frac{\gamma}{2}})  } |z_3-z_1|^{2(\Delta_2+\Delta_{-\frac{\gamma}{2}} -\Delta_3 -\Delta_1  )} \\ &\times |z_3-z_2|^{2 (\Delta_1+\Delta_{-\frac{\gamma}{2}}-\Delta_3-\Delta_2)} G\left ( \frac{(z-z_1)(z_2-z_3)}{ (z-z_3) (z_2-z_1)}  \right )  \label{confinv}
\end{align}
where $\Delta_{-\frac{\gamma}{2}}= -\frac{\gamma}{4} (Q+\frac{\gamma}{4})$ and $\Delta_l= \frac{\alpha_l}{2} (Q- \frac{\alpha_l}{2})$.
We can recover %$C$ and 
$G(z)$ as the following limit
\begin{align}
G(z)&= %\lim_{z_1 \rightarrow 0,\:  z_2 \rightarrow 1, \: z_3 \rightarrow \infty } 
\lim_{ z_3 \to\infty }|z_3|^{4 \Delta_3} \langle    V_{-\frac{\gamma}{2}}  (z)  V_{\alpha_1}(0)  V_{\alpha_2}(1) V_{\alpha_3}(z_3)  \rangle   \label{Glimit}
\end{align}
and this becomes in terms of multiplicative chaos
  $$
  G(z)=|z|^{\frac{\gamma \alpha_1}{2}}   |z-1|^{\frac{\gamma \alpha_2}{2}}  \mathcal{T}(z) 
  $$ 
  where $ \mathcal{T}(z)$ is given by 
\begin{align}\label{Tdefi}
  \mathcal{T}(z) &=%A(-\frac{_\gamma}{^2},\alpha_1,\alpha_2,\alpha_3) 
  const. \mu^{-s+\frac{1}{2}}  \gamma^{-1}\Gamma(s-\hf)  \E [R(z)^{\hf-s}]
 \end{align}
and
\begin{equation*}
R(z)= %e^{\frac{b\gamma^2}{2}}
 \int e^{\gamma X(x)- \frac{\gamma^2}{2} \E[X(x)^2] } \frac{|x-z|^{\frac{\gamma^2}{2}}}{ |x|^{\gamma \alpha_1}  |x-1|^{\gamma \alpha_2}  }  \hat{g}(x)^{1+\frac{\gamma^2}{8} - \frac{\gamma}{4} \sum_l \alpha_l } dx.
\end{equation*}
In particular one has
\begin{equation}\label{T(0)}
\mathcal{T}(0)=C(\alpha_1-\frac{_\gamma}{^2},\alpha_2,\alpha_3).
\end{equation}
A little algebra then shows that the BPZ equation becomes   the hypergeometric equation for  $\mathcal{T}$ 
\begin{equation}\label{hypergeo}
z(1-z)\partial_{zz}^2 \mathcal{T}(z)+  (a_1+a_2z)\partial_z \mathcal{T}(z) +a_3 \mathcal{T}(z)=0
\end{equation}
where $a_i$ are are explicit  constants depending on $\alpha_i$ and $\gamma$. By analysing positive solutions to the hypergeometric equation  one finds that 
\begin{equation*}
\mathcal{T}(z)=\lambda_1 |F_{-}(z)|^2 +\lambda_2 \text{Re}(F_{-}(z)F_{+}(z))+\lambda_3 \text {Im}(F_{-}(z)F_{+}(z)) +\lambda_4  |F_{+}(z)|^2
\end{equation*}
where $\lambda_i \in \R$ and $F_\pm$ are definite hypergeometric functions. It remains to determine the coefficients $\lambda_i$. One has from the asymptotics of  $F_\pm$ as $z\to 0$
\begin{align*}
 \mathcal{T}(z) = \lambda_1+\lambda_2 \text{Re} (z^{a})+ \lambda_3 \text{Im} (z^{a}) + \lambda_4 |z|^{2a} + o(|z|^{2a}  ) 
\end{align*}
where $a=\hf \gamma(Q-\alpha_1)$. On the other hand we may study the explicit expression \eqref{Tdefi} as $z\to 0$ by multiplicative chaos techniques. This gives %($2(1-c)=\gamma(Q-\alpha_1)$)
\begin{equation}
\mathcal{T}(z)= C(\alpha_1-\frac{\gamma}{2},\alpha_2,\alpha_3) -\mu  \frac{\pi}{  l(-\frac{\gamma^2}{4}) l(\frac{\gamma \alpha_1}{2})  l(2+\frac{\gamma^2}{4}- \frac{\gamma \alpha_1}{2}) } C(\alpha_1+\frac{\gamma}{2}, \alpha_2,\alpha_3) |z|^{2a}+ o(|z|^{2a})\label{z=0exp}.
\end{equation}
where $l(x)=\frac{\Gamma (x)}{\Gamma(1-x)}$. 
We conclude that $\lambda_1=C(\alpha_1-\frac{\gamma}{2},\alpha_2,\alpha_3),$
 $\lambda_2=\lambda_3=0$ and 
 $$\lambda_4=-\mu  \frac{\pi}{  l(-\frac{\gamma^2}{4}) l(\frac{\gamma \alpha_1}{2})  l(2+\frac{\gamma^2}{4}- \frac{\gamma \alpha_1}{2}) }  C(\alpha_1+\frac{\gamma}{2}, \alpha_2,\alpha_3).$$
From this we can deduce the following corollary on the 3 point structure constants (this argument is   called Teschner's trick \cite{teschner} in the literature):
%\begin{corollary}\label{3pointconstant}
\begin{equation*}
\frac{C(\alpha_1+\frac{\gamma}{2},\alpha_2,\alpha_3)}{C(\alpha_1-\frac{\gamma}{2},\alpha_2,\alpha_3)}  
= - \frac{1}{\pi \mu}\frac{l(-\frac{\gamma^2}{4})  l(\frac{\gamma \alpha_1}{2}) l(\frac{\alpha_1\gamma}{2}  -\frac{\gamma^2}{4})  l(\frac{\gamma}{4} (-\alpha_1+\alpha_2+\alpha_3- \frac{\gamma}{2}) )   }{l( \frac{\gamma}{4} (\alpha_1+\alpha_2+\alpha_3-\frac{\gamma}{2} - 2Q)  ) l( \frac{\gamma}{4} (\alpha_1+\alpha_2-\alpha_3-\frac{\gamma}{2} ))  l( \frac{\gamma}{4} (\alpha_1+\alpha_3-\alpha_2-\frac{\gamma}{2} )) }. 
\end{equation*}
Similarly the other BPZ equation gives rise to a formula for $\frac{C(\alpha_1+\frac{2}{\gamma},\alpha_2,\alpha_3)}{C(\alpha_1-\frac{2}{\gamma},\alpha_2,\alpha_3)}$.
Our derivation of these identities holds only if the Seiberg bounds are satisfied i.e. $ \alpha_1< Q-\frac{\gamma}{2}$ and $\sum_l \alpha_l>2Q+\frac{\gamma}{2}$ and similarly for the other identity. If one makes the assumption that 
$C( \alpha_1, \alpha_2, \alpha_3)$  %\end{corollary}
is an analytic function in the $\alpha_i$ and these identities hold then one may derive an explicit formula for 3 point structure constants, the DOZZ formula. %Thus for our probabilistic approach to the Liouville model the main challenge is to establish analyticity of the 3 point structure constants.

 The DOZZ conjecture and the representation content of Liouville theory are signals of its integrability. They remain the main challenge for the probabilistic approach to the Liouville model. %They depend on analyticity of the vertex operator correlation functions in the $\alpha_i$ (as well as $\gamma$). This is a nontrivial problem. Indeed, the correlations are {\it not} analytic for the regularised theory! Analyticity, if true, only appears in the limit as the regularization is removed. A similar phenomenon appears in the theory of multiplicative chaos on the circle \cite{fb}. There the moments $\E M_\gamma(S^1)^{-s}$ are believed to be explicit (so called Fyodorov-Bouchaud formula) and meromorphic in $s$ and $\gamma$. However, the proof of this is hard as the moments of the regularised chaos are not analytic.   What happens most likely  is that the regularisation produces non-analytic terms that disappear in the limit. A similar phenomenon happens in a toy model where one only keeps the radial part of the chaos measure: let  $Z(t)=\int_0^t e^{B_s-\mu s}ds$ (so called "Asian Option" model. %\cite{yor}

\section{References}

\renewcommand{\refname}{}    %%%% for this example
%\vspace*{-36pt}              %%%% file only!

\end{document}